\documentclass[12pt,reqno]{amsart}

\usepackage{amssymb, amscd, amsthm, mathrsfs, mathtools}
\usepackage{url,enumerate,colonequals}
\usepackage[margin = 1in]{geometry}
\usepackage{pdfsync}
\usepackage{multirow}
\usepackage[usenames,dvipsnames]{xcolor}
\usepackage{url}
\usepackage{microtype}
\usepackage{tikz-cd}
\usepackage[all]{xy}
\usetikzlibrary{matrix}
\usepackage{float}
\usepackage{bbm}
\usepackage{algorithm}
\usepackage{algpseudocode}
\usepackage{tablefootnote}
\usepackage{amsmath}
\usepackage{nicematrix}
\usepackage[customcolors]{hf-tikz}

\floatstyle{boxed}
\restylefloat{figure}

\newcommand{\F}{\mathbb{F}}

\newcommand{\Z}{\mathbb{Z}}
\newcommand{\PP}{\mathbb{P}}

\newcommand{\sC}{\mathsf{C}}
\newcommand{\sR}{\mathsf{R}}
\newcommand{\sE}{\mathsf{E}}

\newcommand{\Ga}{\alpha}

\newcommand{\Gl}{\lambda}
\newcommand{\ep}{\epsilon}
\newcommand{\cR}{\mathcal{R}}

\def \H {{\mathcal H}}

\newcommand{\calO}{\mathcal{O}}
\newcommand{\calP}{\mathcal{P}}

\newcommand{\sP}{\mathscr{P}}

\newcommand{\cB}{\mathcal{B}}

\newcommand{\tcB}{\tilde{\mathcal{B}}}
\newcommand{\tO}{\tilde{O}}

\def\Gd{{\delta}}
\def\cP{{\mathcal P}}
\def\fP{{\mathfrak P}}

\def\fp{{\mathfrak p}}
\def\fb{{\mathfrak b}}

\def\fl{{\mathfrak l}}
\def\fg{{\mathfrak g}}
\newcommand{\cL}{\mathcal{L}}

\newcommand{\PGL}{\mathrm{PGL}}

\newcommand{\GS}{\mathrm{GS}}

\newcommand{\dist}{\mathrm{dist}}

\def\bv{{\bf v}}
\def\bu{{\bf u}}
\def\rRS {\mbox{Reed-Solomon\ }}
\def\rAG {\mbox{algebraic geometry\ }}

\DeclareMathOperator{\Gal}{\mathrm{Gal}}

\DeclareMathOperator{\Tr}{\mathrm{Tr}}

\DeclareMathOperator{\Aut}{\mathrm{Aut}}

\DeclareMathOperator{\Char}{\mathrm{Char}}

\DeclareMathOperator{\Ker}{\mathrm{ker}}
\DeclareMathOperator{\GFFT}{\mathrm{G}{\text{-}}\mathrm{FFT}}
\DeclareMathOperator{\GIFFT}{\mathrm{G}{\text{-}}\mathrm{IFFT}}

\DeclareMathOperator{\ev}{\mathrm{ev}}
\DeclareMathOperator{\RS}{\mathrm{RS}}
\DeclareMathOperator{\AG}{\mathrm{AG}}
\DeclareMathOperator{\IRS}{\mathrm{IRS}}

\DeclareMathOperator{\bx}{\mathbf{x}}

\DeclareMathOperator{\by}{\mathbf{y}}
\DeclareMathOperator{\be}{\mathbf{e}}

\DeclareMathOperator{\bi}{\mathbf{i}}
\DeclareMathOperator{\bj}{\mathbf{j}}

\DeclareMathOperator{\br}{\mathbf{r}}
\DeclareMathOperator{\bc}{\mathbf{c}}

\DeclareMathOperator{\bw}{\mathbf{w}}
\def\fC{\mathfrak{C}}
\def\x{\times}

\newtheorem{theorem}{Theorem}[section]
\newtheorem{lemma}[theorem]{Lemma}
\newtheorem{corollary}[theorem]{Corollary}

\newtheorem{Maintheorem}{\bf Main Theorem}

\theoremstyle{definition}
\newtheorem{definition}[theorem]{Definition}

\newtheorem{example}[theorem]{Example}

\newtheorem{remark}[theorem]{Remark}
\newtheorem{problem}[theorem]{Problem}

\numberwithin{equation}{subsection}

\begin{document}
\title[Quasi-linear time decoding]{Quasi-linear time decoding of RS and AG codes for burst errors up to the Singleton bound} 

\author{Songsong Li, Shu Liu, Liming Ma, Yunqi Wan, Chaoping Xing}
\address{}
\email{}

\subjclass[]{}
		
\maketitle	
\begin{abstract}
Despite of tremendous research on decoding Reed-Solomon (RS) and algebraic geometry (AG) codes under the random and adversary substitution error models, few studies have explored these codes under the burst substitution error model. Burst errors are prevalent in many communication channels, such as wireless networks, magnetic recording systems, and flash memory. Compared to random and adversarial errors, burst errors often allow for the design of more efficient decoding algorithms. However, achieving both an optimal decoding radius and quasi-linear time complexity for burst error correction remains a significant challenge.. 
The goal of this paper is to design (both list and probabilistic unique) decoding algorithms for RS and AG codes that achieve the Singleton bound for decoding radius while maintaining quasi-linear time complexity. 

Our idea is to build a one-to-one correspondence between AG codes (including RS codes) and interleaved RS codes with shorter code lengths (or even constant lengths). By decoding the interleaved RS codes with burst errors, we derive efficient decoding algorithms for RS and AG codes. For decoding interleaved  RS codes with shorter code lengths, we can employ either the naive methods or existing algorithms. This one-to-one correspondence is constructed using the generalized fast Fourier transform (G-FFT) proposed by Li and Xing  (SODA 2024). The G-FFT generalizes the divide-and-conquer technique from polynomials to algebraic function fields. More precisely speaking, assume that our AG code is defined over a function field $E$ which has a sequence of subfields $\mathbb{F}_q(x)=E_r\subseteq E_{r-1}\subseteq \cdots\subset E_1\subseteq E_0=E$ such that $E_{i-1}/E_i$ are Galois extensions for $1\le i\le r$. Then the AG code based on $E$ can be transformed into an interleaved RS code over the rational function field $\mathbb{F}_q(x)$.
\end{abstract}

\section{Introduction}

There are three types of substitution errors in coding theory, namely, random  substitution errors, adversary substitution errors and burst substitution errors. This paper focuses on burst substitution errors. Let us simply call burst  errors instead of burst substitution errors.
 In some communication channels such as wireless, magnetic recording and  flash channels, errors caused by physical damage tend to be localized in a consecutive interval. 
 Generally, decoding of burst errors is easier than the decoding of other two types of errors. However, in terms of decodability, error-correcting codes have the same performance for both adversary  errors and burst  errors. More precisely speaking, for a fixed rate, the largest decoding (both unique and list decoding) radii of codes are equal for adversary  errors and burst  errors. For unique decoding, a code can correct errors up to half of the minimum distance for both adversary  errors and burst  errors~\cite{Xing04}. For list decoding,  a code can correct errors up to the Singleton bound for adversary errors as well as burst  errors~\cite{gur19essential}.

Reed-Solomon (RS) codes are a fundamental class of algebraic error-correcting codes that have been extensively studied in theoretical computer science. Owing to their elegant algebraic structures and the availability of efficient algorithms such as the fast Fourier transform (FFT), RS codes exhibit strong error-correction capabilities and support efficient decoding algorithms \cite{jus76,gao03,lin16novel}. Significant research has been devoted to understanding the decodability properties of RS codes \cite{sud97,GSlist} and their variants, including folded RS codes \cite{GW13,Kopparty23,Tamo24,Goy24} and interleaved RS codes \cite{GU08IRS}. Notably, adversarial errors in folded RS codes with a folding length \( m \) correspond to phased burst errors of length \( m \) in the unfolded RS codes, with single burst errors representing a special case.

 
Algebraic geometry (AG) codes are constructed from algebraic function fields over finite fields, with RS codes serving as a special case derived from rational function fields. By utilizing function fields with many rational points, AG codes enable the construction of linear codes over an alphabet whose size is smaller than the block length, which is a main constraint in RS codes. There also has been a lot of work concerned with decoding properties of AG codes \cite{sho99list,GSlist, GX22,beelen2022fast} or folded AG codes \cite{GX12foldedAG, GXY17} to reduce alphabet size in the realm of list decoding. However, decoding AG codes is much more complex and also expensive than RS codes due to the involvement with algebraic function fields.

In this work, we propose a novel approach to investigate the decoding properties of AG codes under the burst error model and develop fast decoding algorithms via fast Fourier transform, departing from the established line of research on list decoding of folded AG or Reed-Solomon (RS) codes. Since RS codes are a special case of AG codes, our findings are also applicable to RS codes.

\subsection{Related work}
Let us start by defining burst errors. A nonzero vector $\be\in\F_q^n$ is called a burst of length $\fl$ if the difference between the first nonzero position $i$ and the last nonzero position $j$ of $\be$ is less than $\fl$, i.e., $j-i<\fl$. 

For a code $C\subset \F_q^n$ of length $n$ defined over a finite field $\F_q$, let $d$ be its minimum distance and $\cR=\log_q |C|/n$ be the code rate. The Singleton bound \cite{joshi58} says that
\[d-1\leq n(1-\cR).\]

\subsubsection{List decoding of burst errors}
 Given a real number $0<\rho<1$ and a positive integer $L>0$,
a code $C\subset\F_q^n$ is said to be $(\rho, L)$-burst list-decodable if for any received word $\br\in\F_q^n$, there are at most $L$ codewords $\bc\in C$ such that $\br-\bc$ is a burst of length $\rho n$. Generally, without constraint on burst errors, if there are at most $L$ codewords within a Hamming radius $\rho n$ from $\br$ in the adversary error model, then $C$ is called $(\rho, L)$-list decodable. For an \rRS code of length $n$ and dimension $k$, we can design the following naive burst-list decoding algorithm. Assume that for a received word $\br$, there is a burst error of length $\fl$ less than or equal to $n-k$. We erase any $n-k$ consecutive positions of $\br$ and then interpolate a polynomial of degree at most $k-1$ from the left $k$ values of $\br$. As there are at most $k+1$ such consecutive positions, the list size  $L$ is $k+1$.  As the univariate interpolation costs $\tilde{O}(n)$\footnote{The notion $\tO(n)=O(n\log^{c} n)$ for some constant integer $c>0$, which denotes quasi linear complexity.} operations \cite{von13}, the total list decoding complexity is $\tilde{O}(n^2)$. 
Roth and Vontobel \cite{roth09} showed that the list decoding radius of an $[n,k]$-linear code for burst errors is bounded by the Reiger bound, i.e., the list decoding radius obeys
\[\fl\le \frac L{L+1}(n-k),\]
where $L$ is the list size. the Reiger bound is asymptotically equal to  the Singleton bound. Let us call it the Singleton bound as well. 
Furthermore, they also showed that \rRS codes can be burst list decoded up to Singleton bound in polynomial time. 
Later, Kolan and Roth \cite{kol13burst} proved that the interleaved \rRS codes (a variant of \rRS code) can be list-decoded with burst errors while attaining the Reiger bound. In \cite{Ding15}, Ding showed that a random code $C\subset \F_q^n$ is $(1-\cR-\ep, O(1/\ep))$-burst list-decodable and the random linear code is $(1-\cR-\ep, O(n))$-burst list-decodable. 

Under adversary error model, Sudan and Guruswami \cite{GSlist} showed that an rate $\cR$ $\RS$ code is $(1-\sqrt{\cR}-\ep,O(1/\ep))$-list-decodable with decoding complexity $O(n^3)$. To improve the decoding radius, Guruswami and Wang \cite{GW13} showed that rate $\cR$ folded RS codes are $(1-\cR-\epsilon,q^{1/\epsilon})$-list decodable. After a series of work \cite{Kopparty23,Tamo24,Goy24}, Chen and Zhang recently \cite{chen24} reduced the output list size to $O(1/\epsilon)$. These works are concerned with existential results or randomized algorithms.  
It is a great challenge to design deterministic list decoding algorithms for radius up to Singleton bound $1-{\cR}-\ep$ with optimal output list size $O(1/\ep))$ as well as optimal quasi-linear time $O_{\epsilon}(n\log n)$.

 As for general AG codes, there are also several works concerned with the list decoding of (folded) AG codes under adversarial model \cite{GSlist, GXY17,GX22,ding14erasure}. However, we are not aware of any work on list decoding for single-burst substitution errors. 
Thus we raise the following problem by asking both decoding radius and time complexity as well.

\begin{problem}\label{prob:p3}
{\it Are AG (including RS) codes  $(1-\cR-\ep, O(1/\ep))$-burst list decodable with deterministic quasi-linear time decoding complexity $O_{\epsilon}(n\log n)$?}
\end{problem}

\subsubsection{Unique decoding of burst errors}
It is well known that the unique decoding radius for adversary errors is half of the minimum distance. This is also true for uniquely decoding burst errors of Reed-Solomon codes \cite{Xing04}. In order to decode beyond half of the minimum distance, one has to seek probabilistic decoding algorithms, i.e., correcting errors with high probability and  miscorrecting errors with small probability. In the case of adversary errors, power decoding was introduced for Reed-Solomon \cite{CD09,10synd,17pd,PD18} and algebraic geometry codes \cite{nie15,puc19improved,PDAG21}. Unfortunately, the miscorrection probability analysis of power decoding is not an easy task. Most of the power decoding algorithms do not have failure probability analysis. Instead, these power decoding algorithms usually provide simulation results. We would like to emphasize that so far there are no power decoding algorithms of both \rRS and \rAG codes achieving the Singleton bound in quasi-linear time. 

Similar to decoding of adversary errors beyond half of the minimum distance, one has to design probabilistic decoding algorithms in order to correct burst errors beyond half of the minimum distance. Compared with power decoding for adversary errors, it is relatively easier to analyze the failure probability for  probabilistic algorithms. 
For  probabilistic unique decoding of \rRS codes with burst errors, a key step is to find the interval of the burst. Then one can use a fast erasure decoding algorithm with erased positions being the burst interval to find the right codeword. An upper bound of the burst error length for probabilistic unique decoding is the Singleton bound $\fl<n-k$. Several works~\cite{92burst,01burst,wu12} studied the problem of identifying the burst interval with miscorrection probability. In \cite{wu12}, given a received word $\br\in\F_q^n$, Wu proposed a novel algorithm to identify the shortest burst interval, hence deriving a probabilistic unique decoding algorithm with decoding radius achieving the Singleton bound and  complexity $O(n^3)$. Moreover, he gave a short and clear analysis of the miscorrection probability of this algorithm. However, Wu's algorithm works only for cyclic Reed-Solomon codes with  evaluation points $\{1, \alpha, \alpha^2,\dots,\alpha^{n-1}\}$ being a cyclic subgroup of $\F_q^*$, where $\alpha\in\F_q*$ has order $n$. One open problem for probabilistic unique decoding of burst errors for \rRS codes is the following.

The unique decodability of \rAG codes from burst errors is also roughly half of the minimum distance. More precisely,  we show in Lemma~\ref{lem:uniqburst} that an \rAG code can only deterministically correct burst errors of length up to $\frac12(d^*+\fg)$, where $d^*$ is the designed distance of \rAG code and $\fg$ is the genus of the underlying function field (note that usually we have $\fg=o(d^*)$). In~\cite{Ren04}, Ren proposed a deterministically burst error-correcting algorithm for \rAG codes based on the Hermitian curves. By exploiting the block structure in Hermitian codes, decoding Hermitian codes can be converted to decoding \rRS codes. Ren's algorithm has quasi-linear decoding complexity $\tO(n)$ using the best decoding algorithm for \rRS codes. This algorithm can only correct  adversary errors far less than half of the designed distance. However, it performs better for burst errors, but no explicit error-correcting radius was given. Later, Wang et al. \cite{Kan15} refined Ren's work to obtain a burst error-correcting algorithm that can correct more errors, but with higher time complexity $\tO(n^{2})$.

Again, we raise the following problem for \rAG codes to probabilistic uniquely correct burst errors.
\begin{problem}\label{prob:p5}
{\it Can we design a probabilistic algorithm to decode \rAG codes  for burst errors  up to the Singleton bound with  quasi-linear  complexity $\tilde{O}(n)$? }
\end{problem}

\subsection{Our results} The main goal of this paper is to provide affirmative answers to Problems 1.1 and 1.2 asked in Subsection 1.1. In summary, in this paper we show that {\it under the burst error model, both Reed-Solomon and algebraic geometry codes can be (probabilistic uniquely and list) decoded up to the Singleton bound with quasi-linear complexity. } The detailed results are summarized in the following three main theorems.

\begin{Maintheorem} 
For any real number $\epsilon >0$, the \rRS code $\RS[n,k]$ of length $n$ and dimension $k$ is $(1-\cR-\epsilon,O(1/\epsilon))$-burst list-decodable, where $\cR=k/n$ is the code rate. Furthermore, if $n$ is an $O(1)$-smooth \footnote{If $n$ has a factorization $n=\prod_{i=1}^rp_i$ such that each factor $p_i\leq B$, then we say that $n$ is $B$-smooth. Particularly, if $B=O(1)$ is a constant, we say that $n$ is $O(1)$-smooth} integer, then $\RS[n,k]$ can be list decoded in time $O(\frac 1{\epsilon} n\log n)=O_{\ep}(n\log n)$ by Algorithm 1 given in Subsection 4.1.
\end{Maintheorem}

In \cite{roth09}, it was also shown that \rRS codes are $(1-\cR-\epsilon,O(1/\epsilon))$-burst list-decodable in polynomial time. As we mentioned in the previous subsection, the naive algorithm gives complexity $\tO(n^2)$. Our Main Theorem 1 reduces complexity from quadratic time $\tO(n^2)$ to quasi-linear time $O(\frac 1{\epsilon} n\log n)=O_{\ep}(n\log n)$.

\begin{Maintheorem}
For any real number $\ep>0$ and any integer $e\geq 1$, $\RS[n,k]$ can uniquely correct burst errors of length less than $n(1-\cR-\ep)$ with miscorrection probability at most $\ep/q^{e-1}$. Moreover, if $n$ is an $O(1)$-smooth integer, then probabilistic unique decoding can be performed in time $O(\frac{e}{\ep}n\log n)$ by Algorithm 2  given in Subsection 4.2.
\end{Maintheorem}

To make Wu's algorithm work \cite{wu12}, one requires that the evaluation point set of the Reed-Solomon code forms a multiplicative cyclic group. Our result removes this requirement. In addition, we reduce the time complexity of Wu's algorithm from  $O(n^3)$ to quasi-linear time $O_\ep(n\log n)$.

\begin{Maintheorem}\label{thm:mainthm3}
Assume $E/\F_q(x)$ is an algebraic function field extension of degree $m$ satisfying the G-FFT assumptions in Section~\ref{sec:IRS}. Let $\fC(\lambda P^{(0)}_{\infty},\calP)$ be the \rAG code of length $N$ based on $E$ defined as in Equation~\eqref{eq:agc}. Then, for any real number $\ep>0$,
\begin{enumerate}
\item $\fC(\lambda P^{(0)}_{\infty},\calP)$ is $\left(1-\cR-\fg/N-\epsilon, O(1/\epsilon)\right)$-burst list-decodable, where $\cR$ is the rate of the code and $\fg$ is the genus of $E$ (note that usually we have $\fg=O(\epsilon N)$ or $\fg=o(N)$, thus the decoding radius achieves the Singleton bound). Furthermore, if both $N$ and $m$ are $O(1)$-smooth, then $\fC(\lambda P^{(0)}_{\infty},\calP)$ can be list  decoded in quasi linear time $O(\frac1{\ep}N\log N)$.
\item let $e$ be any positive integer satisfying $m\ep\leq q^{e-1}$, then $\fC(\lambda P^{(0)}_{\infty},\calP)$ can uniquely correct burst errors of length less than $N(1-\cR-\fg/N-\ep)$ with miscorrection probability at most $m\ep/q^{e-1}$.  Furthermore, if both $N$ and $m$ are $O(1)$-smooth, then $\fC(\lambda P^{(0)}_{\infty},\calP)$ can be decoded in quasi-linear time $O(\frac e{\ep}N\log N)=O(\frac 1{\ep}N\log^2 N)$ (note that in the case of Reed-Solomon code or Hermitian code, $e$ can be constant, thus the decoding complexity is $O(\frac 1{\ep}N\log N)$).
\end{enumerate}

\end{Maintheorem}

We want to emphasize that the G-FFT assumptions in Main Theorem \ref{thm:mainthm3} are usually easy to be satisfied as demonstrated by \cite{encAG24} and our examples. We present examples of \rAG codes in Section 5 to illustrate our Main Theorem \ref{thm:mainthm3}. Examples include the well-known Hermitian curves and the Hermitian tower. Besides, we also give an example of \rAG subcodes based on the third function field in the Garcia-Stichtenoth tower \cite{gs95tower}.

\subsection{Overview of techniques}
Our main technique is to make use of the fast Fourier Transform (FFT for short) to build a one-to-one correspondence between \rAG codes and interleaved $\RS$ code ($\IRS$ code for short) (refer to Subsection~\ref{sec:AG&IRS} for the detailed correspondence). Under this correspondence, an $\fl$-burst in $\F_q^n$ corresponds to a shorter $\lfloor\frac{\fl}{m}\rfloor+2$-burst matrix $\sE$ in $\F_q^{m\times (n/m)}$. Then we can convert the burst decoding problem in \rRS or \rAG codes to a shorter and easier decoding problem in $\IRS$ codes. For a better understanding of our main techniques, let us first introduce the FFT and formal definitions of $\RS$ and $\IRS$ codes. 

\subsubsection{The fast Fourier Transform} The discrete Fourier transform (DFT) of length $n$ over a finite field $\F_q$ is a transform from $n$ coefficients of a polynomial $f(x)$ over $\F_q$ of degree less than $n$ to $n$ evaluations of $f(x)$ at $n$ distinct elements in $\F_q$ (hence $n\leq q$). The inverse discrete Fourier transform (IDFT) is the reverse process from $n$ evaluations to $n$ coefficients. The FTT and inverse FFT (IFFT) are fast algorithms that make the DFT and IDFT can be performed in time $O(n\log n)$. We refer to \cite{LX24} for an overview of FFT. 

In \cite{LX24}, the authors introduced a unified framework for known multiplicative FFT \cite{Pollard71} and additive FFT \cite{lin2014novel} by leveraging the automorphism group of the rational function field and Galois theory. They termed this framework G-FFT to distinguish it from existing FFTs. In \cite{encAG24}, the authors extended the G-FFT framework from rational function fields to algebraic function fields. They demonstrated that AG codes constructed from algebraic function fields satisfying the G-FFT assumptions (see \S~\ref{sec:fastEncAG}) can be encoded in quasi-linear time. In this work, we utilize the theoretical foundations established in \cite{encAG24,LX24} to investigate burst decoding for RS and AG codes. Notably, Bordage et al. \cite{AGIOPP} also established fast encoding of AG codes via automorphism groups. While their approach is based on more general Riemann-Roch spaces, \cite{encAG24} focuses specifically on one-point Riemann-Roch spaces and provides more concrete examples of AG codes.

\subsubsection{Interleaved \rRS codes}\label{sec:preIRS} For the sake of simplicity, let us discuss the correspondence between the simplest algebraic geometry codes, i.e., Reed-Solomon codes and Interleaved \rRS codes. We will then briefly mention the  correspondence between general algebraic geometry codes and Interleaved \rRS at the end of this subsection.
Let $\calP=\{\Ga_1,\Ga_2,\dots,\Ga_n\}\subseteq \F_q$ be a set of $n$ pairwise distinct elements. For a positive integer integer $k\leq n$, let $\F_{q}[x]_{<k}$ be the polynomial space of degree less than $k$. The Reed-Solomon code is defined as
\[\RS[n,k]:=\{\ev_{\calP}(f)=(f(\Ga_1), \ldots,f(\Ga_n))\mid f\in\F_{q}[x]_{<k}\}.\]
We usually call $\calP$ the set of evaluation points and $\ev_{\calP}(f)$ the multipoint evaluation (MPE for short) of $f$ at $\calP$.

Assume $m$ is a positive integer. An $m$-interleaved RS (IRS for short) code  is the direct sum of $m$ RS codes $\{\RS[n,k_i]\}_{i=1}^m$ having the same set of evaluation points, i.e.,
\[{\IRS}_m(n,k_1,\ldots,k_m)=\left\{ \left(\begin{matrix}
\ev_{\calP}(f_1)  \\ 
\ev_{\calP}(f_2) \\
\vdots       \\
\ev_{\calP}(f_m) \end{matrix} \right) \mid f_i(x)\in\F_q[x]_{<k_i},\ i=1,\dots,m\right\}\subseteq\F_q^{m\times n}.\]
A homogeneous $\IRS$ code is an $\IRS$ code whose $m$ constituent \rRS codes have the same dimension $k_i=k$. In this case, we denote it by $\IRS_m(n,k)$. Otherwise, we call it a heterogeneous $\IRS$ code. In this work, for a matrix $M\in\F_q^{m\times n}$, we will denote the $i$-row of $M$ by $M(i)$ and the $j$-th column of $M$ by $M[j]$ for $i\in[m]$ and $j\in[n]$, respectively. 

The idea for burst decoding \rRS and \rAG codes using G-FFT is similar. Let us first explain our idea for decoding \rRS codes. We illustrate our main technique through the multiplicative FFT without using automorphism groups and field extensions. 
For simplicity, let us assume that $n=2^r$ is a factor of $q-1$. Let $\calP=\{1,\alpha,\dots,\alpha^{n-1}\}$ be a subgroup of $\F_q^*$ of order $n$. Let $m_s=2^s$ for $s\le r$ and $n_s=n/m_s=2^{r-s}$. Then the $m_s$-th power maps $\calP$ to a cyclic group of order $n_s$ is given by 
\[(\cdot)^{m_s}:\ \calP\rightarrow \calP_s;\ \alpha^i\mapsto \alpha^{im_s}.\]
Let $\{\beta_1,\dots,\beta_{n_s}\}$ denote the elements in $\calP_s$. Then every $\beta_j$ has $m_s$ preimages in $\calP$, denoted by $\calP[j]=\{\Ga_{1,j},\dots,\Ga_{m_s,j}\}$. Moreover, a polynomial $f(x)$ of degree less than $k$ can be written as a linear combination of $1,x,\dots,x^{m_s-1}$ with coefficients being polynomials of degree less than $k_s:=\lceil k/m_s \rceil$, i.e.,
\[f(x)=f_0(x^{m_s})+f_1(x^{m_s})x+\dots+f_{m_s-1}(x^{m_s})x^{m_s-1},\ \text{all}\ \deg(f_i)<k_s.\]
Let $\RS[n,k]$ be the \rRS code defined over $\calP=\bigsqcup_{j=1}^{n_s} \calP[j]$ in the following way
\[\RS[n,k]=\left\{\bc_f:=\left(\ev_{\calP[1]}(f),\ev_{\calP[2]}(f),\dots,\ev_{\calP[n_s]}(f)\right)\mid f(x)\in\F_q[x]_{<k}\right\}.\]
Then $\RS[n,k]$ can be encoded in time $O(n\log n)$ via G-FFT.
Note that for every $1\leq j\leq n_s$,
\[\ev_{\calP[j]}(f)=\ev_{\calP[j]}\left(f_0(\beta_j)+f_1(\beta_j)x+\dots+f_{m_s-1}(\beta_j)x^{m_s-1}\right).\]
By running IFFT of length $n_s$, we can recover the coefficients $\left(f_0(\beta_j), f_1(\beta_j),\dots, f_{m_s-1}(\beta_j)\right)$ from $\ev_{\calP[j]}(f)$. Therefore, for any codeword $\bc_f\in\RS[n,k]$, we can first fold it with length $m_s$ and then run IFFT over all columns $\ev_{\calP[j]}(f)$, obtaining the following correspondence and a matrix $\sC_f$ given below, 
\begin{equation}\label{eq:convert}
\small
\begin{split}
\tau:
\bc_f=&(\ev_{\calP[1]}(f),\ev_{\calP[2]}(f),\ldots,\ev_{\calP[n_s]}(f))\xrightleftharpoons[m_s-unfolding]{m_s-folding}\left(\begin{matrix}
f(\Ga_{1,1})           &         f(\Ga_{1,2})      & \cdots  & f(\Ga_{1,n_s})  \\ 
f(\Ga_{2,1})            &         f(\Ga_{2,2})      & \cdots  & f(\Ga_{2,n_s})  \\ 
\vdots       & \vdots       & \cdots   & \vdots \\
f(\Ga_{m_s,1})          &         f(\Ga_{m_s,2})     & \cdots  & f(\Ga_{m_s,n_s})  \\  \end{matrix}\right) \\
&\xrightleftharpoons[column\ FFT]{column\ IFFT}\left(\begin{matrix}
f_0(\beta_1)           &          f_0(\beta_2)    & \cdots  & f_0(\beta_{n_s})\\ 
f_1(\beta_1)           &          f_1(\beta_2)    & \cdots  & f_1(\beta_{n_s})\\ 
\vdots       & \vdots       & \cdots   & \vdots \\
f_{m_s-1}(\beta_1)           &          f_{m_s-1}(\beta_2)    & \cdots  & f_{m_s-1}(\beta_{n_s})\\  \end{matrix}\right)=\sC_f.
\end{split}
\end{equation}
Note that the $i$-row of $\sC_f$ is $\sC_f(i)=\ev_{\calP_s}(f_{i-1}(x))\in\RS[n_s,k_s]$. Thus $\sC_f$ is indeed the interleaved \rRS code  $\IRS_{m_s}(n_s,k_s)$. Conversely, given a codeword $C_f$ in $\IRS_{m_s}(n_s,k_s)$, by the inverse map $\tau^{-1}$, we can get a codeword $\bc_f\in\RS[n,k]$. Similarly, under the map $\tau$, any vector $\br\in\F_q^n$ can be converted to a matrix $\sR\in\F_q^{m_s\times n_s}$ which may not be a codeword in $\IRS_{m_s}(n_s,k_s)$ as the corresponding polynomials in each row may have degree $\geq k_s$. Particularly, an $\fl$-burst $\be\in\F_q^n$ is mapped to a matrix $\sE$ which has at most $\lfloor\fl/m_s \rfloor+2$ nonzero columns (which we call an $\lfloor\fl/m_s \rfloor+2$-burst in $\F_q^{m_s\times n_s}$). More precisely, assume the burst interval of $\be$ is $[a,b]$ and $a=m_sq_a+r_a$, $b=m_sq_b+r_b$, where $0\leq r_a, r_b<m_s$. Then
\[
\small
\begin{split}
\be=&\begin{pNiceMatrix}[ 
    code-for-first-row = \color{blue},
    code-for-first-col = \color{green},
    xdots/line-style=loosely dotted,
    parallelize-diags = false,
    first-row,
    first-col,
    nullify-dots,]
& 1 & \ldots & a & \dots & b  & \dots & n  \\
& 0 & \ldots &  \Block[fill=red!15,rounded-corners]{1-3}{} \x & \dots & \x  & \dots &0  \\
\CodeAfter
\UnderBrace[shorten,yshift=1.5mm]{last-3}{last-5}{\fl}
\end{pNiceMatrix}\xrightleftharpoons[m_s-unfolding]{m_s-folding}
\begin{pNiceMatrix}[ 
    code-for-first-row = \color{blue},
    code-for-first-col = \color{purple},
    xdots/line-style=loosely dotted,
    parallelize-diags = false,
    first-row,
    first-col,
    nullify-dots,]  
   & 1 & \cdots & q_a & q_a+1 &\cdots & q_b & \cdots &n_s\\
1 & 0 & \cdots &         0 &\Block[fill=red!15,rounded-corners]{7-2}{}         \x & \cdots   &  \Block[fill=red!15,rounded-corners]{5-1}{}\x    &  \cdots  & 0 \\
\vdots & \vdots & \cdots & \vdots & \vdots & \vdots & \vdots    & \vdots  &\vdots \\
r_a &     0 & \cdots & \Block[fill=red!15,rounded-corners]{5-1}{}         \x &         \x & \cdots  &          \x  &   \cdots   & 0 \\
\vdots & \vdots & \cdots & \vdots & \vdots & \vdots & \vdots    & \vdots  &\vdots \\
r_b &     0 & \cdots &         \x &         \x & \cdots  &          \x  &   \cdots   & 0 \\
\vdots &      0 & \cdots &  \vdots & \vdots & \vdots & \vdots   &  \vdots  &\vdots \\
m_s&   0 & \cdots &         \x &         \x & \cdots  &          0  &   \cdots   & 0 \\
\CodeAfter
\UnderBrace[shorten,yshift=1.5mm]{last-3}{last-6}{\leq\lfloor\fl/m_s \rfloor+2}
\end{pNiceMatrix}\\
& \\
&\xrightleftharpoons[column\ FFT]{column\ IFFT}
\begin{pNiceMatrix}[ 
    code-for-first-row = \color{blue},
    code-for-first-col = \color{purple},
    xdots/line-style=loosely dotted,
    parallelize-diags = false,
    first-row,
    first-col,
    nullify-dots,]  
   & 1 & \cdots & q_a & q_a+1 &\cdots & q_b & q_b+1 &\cdots &n_s\\
1 & 0 & \cdots &\Block[fill=red!15,rounded-corners]{7-4}{}          \x &        \x & \cdots   & \x  & 0  &  \cdots  & 0 \\\vdots & \vdots & \cdots & \vdots & \vdots & \vdots & \vdots   & 0 & \vdots  &\vdots \\
r_a &     0 & \cdots &         \x &         \x & \cdots  &          \x  & 0 &  \cdots   & 0 \\
\vdots & \vdots & \cdots & \vdots & \vdots & \vdots & \vdots    & \vdots  &\vdots \\
r_b &     0 & \cdots &         \x &         \x & \cdots  &          \x  & 0 &  \cdots   & 0 \\
\vdots &      0 & \cdots &  \vdots & \vdots & \vdots & \vdots   & 0 & \vdots  &\vdots \\
m_s&   0 & \cdots &         \x &         \x & \cdots  &          \x  & 0 &  \cdots   & 0 \\
\CodeAfter
\UnderBrace[shorten,yshift=1.5mm]{last-3}{last-6}{\leq\lfloor\fl/m_s \rfloor+2}
\end{pNiceMatrix}=\sE\\
& \\
\end{split}
\]

For a corrupted codeword $\br=\bc_0+\be$ by a burst $\be$, we get a corrupted $\IRS$ codeword $\sR=\sC_0+\sE$ by a burst $\sE$ under the above map $\tau$. Note that $\sE[j]\neq 0$ if and only if $\sC_0[j]\neq\sR[j]$ for $1\leq j\leq n_s$. Thus, we convert the decoding problem of $\RS[n,k]$ to the decoding problem of shorter $\IRS_{m_s}[n_s,k_s]$. If we can get a codeword $\sC\in\IRS_{m_s}(n_s,k_s)$ from $\sR$ such that $|\{j\mid \sR[j]\neq\sC[j]\}|\leq \lfloor\fl/m_s \rfloor+2$, by the inverse map, we can get a codeword $\bc=\tau^{-1}(\sC)\in\RS[n,k]$ such that $\br-\bc$ is an $(\fl+2m_s-r_s)$-burst, where $r_s=\fl \mod m_s$. In the case of list decoding, we only put the the corrected codeword $\bc=\tau^{-1}(\sC)$ such that $\tau^{-1}(\sC)-\br$ is an $\fl$-burst into our decoding list $\mathfrak{L}$; in the case of probabilistic unique decoding, we will find $\sC_0$ with certain miscorrection probability and compute $\bc_0=\tau^{-1}(\sC_0)$.

For burst decoding of $\IRS$ code, as the errors in each row $\sE(i)$ belong to the same burst interval (i.e., $[q_a,q_b]$ in the above example), its decoding capacity is determined by its constituent \rRS code $\RS[n_s,k_s]$. Specifically, 
\begin{itemize}
\item[(i)] for list decoding, since $\RS[n_s,k_s]$ is $O(1-k_s/n_s,O(n_s))$-list decodable, $\IRS_{m_s}(n_s,k_s)$ can list decode burst errors 
$\lfloor\fl/m_s \rfloor+2\leq n_s-k_s$, i.e., $\fl\leq n-k-2m_s$ with list size $<n_s$. By taking $n_s=n/m_s=2/\ep$, we have $\RS[n,k]$ is $O(1-\cR-\ep,O(1/\ep))$-list decodable. 
\item[(ii)] for probabilistic unique decoding, the corresponding $\IRS_{m_s}(n_s,k_s)$ is defined over a cyclic group, which means that its constituent $\RS[n_s,k_s]$ can be decoded by Wu's algorithm (see Lemma~\ref{lem:UniDec}). Thus we can run Wu's algorithm for all rows of $\sR$ independently to get the correct $\sC_0\in\IRS_{m_s}(n_s,k_s)$. By carefully choosing parameters of burst length $\fl$ and redundancy $n_s-k_s$, we derive a unique decoding algorithm for $\RS[n,k]$ which can correct burst errors up to Singleton bound with certain miscorrection probability.
\end{itemize}
The above $\RS[n,k]$ is defined on a cyclic evaluation-point set. So it is suitable for Wu's algorithm for probabilistic unique decoding which would have better decoding performance than our probabilistic unique decoding strategy as above (ii). However, this example using multiplicative FFT nicely illustrates our main technique in burst decoding \rAG codes using G-FFT. 


In the case of algebraic geometry codes or \rRS defined over $\F_{2^r}$, correspondence based on multiplicative groups may not work. Instead, we have to build the correspondence via additive subgroups of $\F_q$. To solve this problem, we unify the multiplicative and additive approaches by considering automorphism groups of function fields. This unified technique has been pointed out in \cite{LX24} by a short remark. In section 4, we will give a precise description and show the extension of places in the corresponding function field extension that will be used to define \rRS codes. Let us briefly explain the idea and the reader may refer to Sections 3 and 4 for the details. At a higher level of view, the above technique for decoding of Reed-Solomon codes can be viewed as a special case under the following framework. Assume that our algebraic geometry code is based on the function field $E$. We first find a subfield tower $\F_q(x)=E_r\subseteq E_{r-1}\subseteq \cdots\subset E_1\subseteq E_0=E$ such that $E_{i-1}/E_i$ are Galois extensions for $1\le i\le r$. Then we establish the G-FFT over function field $E$. By a similar technique as in Equation \eqref{eq:convert}, the algebraic geometry code based on $E$ can be converted into an interleaved Reed-Solomon code based on $\F_q(x)$ via G-FFT. Therefore, burst decoding of \rAG codes (including the \rRS codes) can be reduced to burst decoding of shorter interleaved \rRS codes. In particular, if $E$ is the rational function field $\F_q(z)$ and the Galois groups ${\rm Gal}(E_{i-1}/E_i)\simeq\Z_2$ are multiplicative groups of $\F_q^*$, then we come to the case of Reed-Solomon codes described above. 

\subsection{Organization}
The paper is organized as follows. In Section 2, we introduce some necessary preliminaries on algebraic function fields, the definition of \rAG codes, and burst decoding \rRS and interleaved \rRS codes. Section 3 presents the transform between fast encodable algebraic geometry codes with interleaved Reed-Solomon codes. Under this framework, we show in Section 4 that \rRS codes, as an example of fast encodable AG codes, can be probabilistic unique and list decoded up to Singleton bound in burst error model. Moreover, two quasi-linear time algorithms for list and probabilistic unique decoding RS codes are included as well. In Section 5, we first show that the burst decoding results of RS codes can be extended to fast encodable AG codes. Then we show that the AG codes based on Hermitian curves and Hermitian towers are examples of our main Theorems. Moreover, we also provide a subcode example of the AG code based on the third Garcia-Stichtenoth tower that can be quasi-linear time decoded with burst errors less than the Singleton bound.
\section{Preliminary}

\subsection{Algebraic function fields}
An algebraic function field over $\F_q$ in one variable is a field extension $E\supset \F_q$ such that $E$ is an algebraic extension of the rational function field $\F_q(x)$ for some transcendental element $x$ over $\F_q$ (see \cite{sti09}). We always denote $\F_q(x)$ by $F$.  In the following, we say $E/\F_q$ is an algebraic function field with the assumption that $\F_q$ is the full constant field of $F$. 

\subsubsection{Places and the evaluations at places} Let $\PP_E$ denote the set of places of $E$. For each place $P\in \PP_E$, one can define a discrete valuation $\nu_{P}$ which is a map from $E$ to $\Z\cup\{\infty\}$ satisfying certain properties. The integral  ring of $P$ is given by ${\calO}_{P}:=\{z\in E:\; \nu_{P}(z)\ge 0\}$. Then ${\calO}_{P}$ is a local ring and its unique maximal ideal is $\{z\in E\mid \nu_P(z)>0\}$. With abuse of notation, we still denote this ideal by $P$. The residue class field $E_{P}:={\calO}_{P}/P$ is a field extension of $\F_q$, and the degree of field extension $[E_{P}: \F_q]$ is called the degree of $P$, denoted by $\deg({P})$. When $\deg(P)=1$, we say that $P$ is $\F_q$-rational (or simply rational if there is no confusion). In this case, for any $f\in E\cap\calO_{P}$, $f(P):=f \bmod P\in E_{P}=\F_q$. We call $f(P)$ is the evaluation of $f$ at $P$. In this work, we mainly consider the rational places in $\PP_E$. 

Particularly, if $E=\F_q(x)$ is the rational function field, then the set of all rational places are one-to-one corresponding to monic linear polynomials plus $1/x$, namely,
\[\begin{split} \{\text{rational\ places\ in}\ &\F_q(x)\}\longleftrightarrow \{x-\Ga \mid \Ga\in\F_q\}\cup \{1/x\},\\
& P_{\Ga} \mapsto (x-\Ga),\ P_{\infty}\mapsto 1/x.\end{split}\]
We may abuse $P=(x-\alpha)$ to denote the zero place of $x-\alpha$ in $F$. The place $P_{\infty}$ corresponding to $1/x$ is the unique pole place of $x$, i.e., $\nu_{P_{\infty }}(x)=-1$. For any $f\in\F_q[x]$, the evaluation of $f$ at a rational place $P_{\Ga}$ is $f(P_{\Ga})=f(\Ga)\bmod P_{\Ga}$.

\subsubsection{The extension of places in function filed extension} Assume $E/F$ is a finite algebraic extension of degree $m$ with the constant field $\F_q$. Assume $P\in\PP_{F}$ is a rational place. A place $\fP\in \mathbb{P}_{E}$ is said to be lying over $P$, written as $\fP\mid P$, if $P\subseteq \fP$. Let $\fP_1, \dots, \fP_t$ be all the places of $E$ lying over $P$. If the number of splitting places lying over $P$ equals the field extension degree, i.e., $t=[E: F]$, then $P$ is said to be \textbf{spliting completely} in $E$. In this case, all $\fP_i\mid P$ are also rational in $E$. If there is a unique rational place in $E$ lying over $P$, then $P$ is said to be \textbf{totally ramified}. A fact is that for every place $P$ of $F$ and any function $f\in F\cap\calO_{P}$, then $f(\fP)=f(P)$ holds for all $\fP\mid P$, namely, $f$ is a constant function restricted on the set $\{\fP\in\PP_E:\ \fP\mid P\}$.

\subsection{Algebraic geometry codes}
Let $E/\F_q$ be a function field of one variable over $\F_q$ of genus $\fg$. 
Choose a set of $N+1$ distinct rational places $P_1,\ldots,P_N, P_{\infty}$ of $E$. Let $\lambda>2\fg-2$ be a positive integer. The one-point Riemann-Roch space associated with $P_{\infty}$ is defined as
\[\cL(\lambda P_{\infty}):=\{f\in E^*\mid \nu_{P_{\infty}}(f)  \geq -\lambda,\ \nu_P(f)\ge 0\, \mbox{if $P\neq P_{\infty}$}\}\cup \{0\}.\]
Then $\cL(\lambda P_{\infty})$ is a finite-dimensional vector space over $\F_q$ which has dimension
\[\dim_{\F_q}(\cL(\lambda P_{\infty})=\lambda+1-\fg.\]
Let $\calP=\{P_{1},P_2,\ldots,P_N\}$. As $P_{\infty}\notin\calP$, then $\cL(\lambda P_{\infty})\subseteq\bigcap_{i=1}^{N}{\calO}_{P_{i}}$. Thus, it is meaningful to define the $\F_q$-linear map
$\ev_{\calP}:\cL(\lambda P_{\infty})\longrightarrow\F_{q}^{N}$ by
$$\ev_{\calP}(f)=(f(P_{1}),f(P_2),\ldots,f(P_{N})) \qquad \hbox{for all} \; f\in\cL(\lambda P_{\infty}).$$
The image of $\ev_{\calP}$ is a linear subspace of $\F_{q}^n$ which is denoted by $\fC(\lambda P_{\infty},\calP)$ and called a one-point {algebraic geometry code} (or {\bf {$\AG$} code}). 
By~\cite[Chapter 2]{sti09}, we know that $\fC(\lambda P_{\infty},\calP)$ {\it is an} $[N,K,d]$-{\it linear code over} $\F_q$ { with}
$$K=\dim_{\F_q}(\cL(\lambda P_{\infty}))=\lambda-\fg+1, \qquad d\ge d^*:=N-\lambda,$$
where $d^*$ is called the designed distance. Particularly, if $E=\F_q(x)$, then $\cL(\lambda P_{\infty})$ is the polynomial space $\F_q[x]_{\le \lambda}$, where $P_{\infty}$ is the unique pole place of $x$. In this case, the \rAG code $\fC(\lambda P_{\infty},\calP)=\RS[n,\lambda+1]$, which has dimension $\lambda+1$ as $\fg(\F_q(x))=0$.

\subsection{Burst error decoding}
In this subsection, we will introduce some known results about decoding $\RS$ and $\IRS$ codes when burst errors occur. Let us first present the definition of bursts for $\RS$ codes.
\begin{definition} ($\fl$-bursts in $\F_q^n$)
For an integer $\fl$ with $0 \leq \fl \leq n$, a word $\be \in \F_q^n$ is called an $\fl$-burst if either $\be = 0$ or the indexes $i$ and $j$ of the first and last nonzero entries in $\be$ satisfy $0 \leq j-i <\fl$. Moreover, the interval $[i,j]$ is called the burst interval of $\be$.
\end{definition}

For $\IRS$ codes, a common error model is the column burst error, where an error position corrupts a whole column simultaneously. The error metric is counted by the number of non-zero columns of the error matrix. Specifically, assume $\sR = \sC+\sE$ is an erroneous received codeword, where $\sC \in{\IRS}_m(n,k_1,\ldots,k_m)$ and $\sE\in \F_q^{m\times n}$. Then the set of error locations in $\sE$ is $\mathbb{E} = \bigcup_{i=1}^m\left\{j\in[n]\mid e_{i,j}\neq 0\right\}$, and the number of errors in $\sE$ is $|\mathbb{E} |$. Particularly, if the nonzero columns in $\sE$ are consecutive, we call $\sE$ a \textbf{burst}.

\begin{definition} ($\fl$-bursts in $\F_q^{m\times n}$)
For an integer $\fl$ with $0 \leq \fl \leq n$, a matrix $\sE \in \F_q^{m\times n}$ is called an $\fl$-burst if either $\sE = 0$ or the indexes $i$ and $j$ of the first and last nonzero columns of $\sE$ satisfy $0 \leq j-i <\fl$. Moreover, the interval $[i,j]$ is called the burst interval of $\sE$.
\end{definition}

In decoding \rRS or \rAG codes with burst errors, finding the burst intervals is enough. In the case of decoding $\RS[n,k]$, given a received word $\br\in\F_q^n$ from the $i$-th position to $i+\fl-1$-th position is corrupted and $\fl\leq n-k$,  we can erase all positions of the received word in the burst interval $[i,i+n-k-1]$. Then the right codeword can be found by erasure decoding algorithm for $\RS$ code. In this work, we will consider both list decoding and probabilistic unique decoding \rRS (resp. \rAG) codes with burst errors. 

\subsubsection{List decoding with burst errors}

For any $\br\in\F_q^n$ (resp. $\sR\in\F_q^{m\times n}$), let $B_{\fb}(\br, \fl)$ (resp. $B_{\fb}(\sR, \fl)$) denote the set of all words which take an $\fl$-burst error from $\br$ (resp. $\sR$). That is,
\[\begin{split}
&B_{\fb}(\br, \fl) :=\{\br + \be \mid \be \in \F_q^n\ \text{is\ an}\ \fl\text{-burst}\},\\
&B_{\fb}(\sR, \fl) :=\{\sR + \sE \mid \sE \in \F_q^{m\times n}\ \text{is\ an}\ \fl\text{-burst}\}
.\end{split}\]

\begin{definition}  Let $L > 1$ be an integer and $\rho \in (0, 1)$. A code $C \subset \F_q^n$ is said to be $(\rho, L)$-burst list-decodable, if for any $\br \in \F_q^n$,  there are at most $L$ codewords $\bc$ in $C$ such that $\br-\bc$ is an $\fl$-burst, i.e.,
\[ |B_{\fb}(\br,\rho n)\cap C| \leq L.\]
\end{definition} 

\begin{lemma}\label{lem:easy}
Assume $m\mid n$ and $u=n/m$. Let $\be\in\F_q^n$  be an $\fl$-burst. By folding $\be$ with length $m$ produces a matrix $\sE$ in $\F_q^{m\times u}$. Then $\sE$ has at most $\lfloor \fl/m\rfloor+2$ consecutive columns are nonzero, i.e., $\sE$ is an $\lfloor \fl/m\rfloor+2$-burst in $\F_q^{m\times u}$.
\end{lemma}
\begin{proof}
The proof is straightforward.
\end{proof}

Similarly, the burst list-decodability of $\IRS$ code is defined as follows.

\begin{definition}  Let $L > 1$ be an integer and $\rho \in (0, 1)$. An $\IRS$ code $\IRS_m(n,k) \subset \F_q^{m\times n}$ is said to be $(\rho, L)$-burst list-decodable, if for any $\sR \in \F_q^{m\times n}$,  there are at most $L$ codewords $\sC$ in $\IRS_m(n,k)$ such that $\sR-\sC$ is an $\fl$-burst in $\F_q^{m\times n}$.
\end{definition}

\subsubsection{Unique decoding with burst errors}

If the set of evaluation point $\calP$ of $\RS[n,k]$ is a cyclic subgroup of $\F_q$, Wu \cite{wu12} gave a novel probabilistic unique decoding algorithm for $\RS$ code that can correct burst errors of length less than $n-k-e$ with miscorrection probability at most $1/q^e$. By a similar proof, we can generalize Wu's algorithm to $\RS$ code defined on cosets of a cyclic group of order $n$, i.e., $\calP\in\F_q^*/\langle \alpha\rangle$, where $\alpha\in\F_q^*$. In Section 4, we will further generalize Wu's algorithm to general $\calP$. 

\begin{lemma}{\cite{wu12}}\label{lem:UniDec}
Let $\alpha\in\F_q^*$ be an element of order $n$. Assume the evaluation set $\calP=\{\xi,\xi\alpha,\dots,\xi\alpha^{n-1}\}$ of $\RS[n,k]$ is a coset of $\F_q^*/\langle \alpha\rangle$. Then $\RS[n,k]$ can correct burst errors of length $\fl< n-k-e$ with miscorrection probability at most $1/q^{e}$. Moreover, it can be decoded in time $O(n\log n)$ by Algorithm 3 in Appendix A if $n$ is an $O(1)$-smooth integer. 
\end{lemma}
\begin{proof}
See Appendix~\ref{app:uniq}.
\end{proof}

The probabilistic unique decoding Algorithm 3 in Appendix A is derived from \cite[Algorithm 1]{wu12}, which always finds a codeword $\bc$ closest to a received word $\br$ in burst error model, i.e., $\bc$ is a codeword such that $\br-\bc$ is the shortest burst. For fast decoding with quasi-linear time, we modify the root finding step of $\Gamma(x)$ via FFT and use the quasi-linear erasure decoding algorithm \cite[Algorithm 4]{lin16novel} to correct $\br$ with a known burst interval.

\section{AG-based Interleaved RS codes}\label{sec:IRS}	
Li et al. \cite{encAG24} generalized the G-FFT \cite{LX24} framework over rational function fields to algebraic function fields, which enables some \rAG codes to be encoded in time $O(N\log N)$ (here $N$ denotes the length of $\AG$ codes). In this section, we use the G-FFT technique in \cite{encAG24} to convert \rAG codes to interleaved Reed-Solomon codes, which will be used to decode AG and RS  codes with burst errors in subsequent sections.

\subsection{$O(N\log N)$ time encodable {$\AG$} codes}\label{sec:fastEncAG}
Let us first briefly recap the G-FFT framework using algebraic function fields in \cite{encAG24}. Let $F=\F_q(x)$ be a rational function field and $P_{\infty}$ be the unique pole place of $x$. Assume $E/F$ is a Galois extension with field extension degree $m=[E: F]$. Suppose that $m=\prod_{i=1}^rp_i$ is a factorization of $m$, where the factors $p_1,\dots,p_r$ are not required to be distinct. To establish G-FFT over $E$, the field extension $E/F$ is expected to satisfy the following assumptions:\bigskip
\\
\fbox{\begin{minipage}{39.5em}
\begin{center}
   {\bf  G-FFT assumptions}
\end{center}
\begin{enumerate}
\item[(P1)] The pole place $P_{\infty}$ of $x$ is totally ramified in $E/F$.
\item[(P2)] There exists a tower of subfields of $E$: $E=E_0\supsetneq E_1\supsetneq\dots \supsetneq E_r=F$ such that $E_{i}/E_{i+1}$ is a Galois extension and the extension degree $[E_{i}: E_{i+1}]=p_{i+1}$ for $0\leq i\leq r-1$.
\item[(P3)] Let $E_{i}=E_{i+1}(y_{i})$ and $P^{(i)}_{\infty}\in\PP_{E_i}$ be the unique place in $E_i$ lying over $P_{\infty}$. Each $y_i$ has a unique pole place $P^{(i)}_{\infty}$ and $\gcd(\nu_{P^{(i)}_{\infty}}(y_i), p_{i+1})=1$. Moreover, $\{1,y_i,\dots,y_i^{p_i-1}\}$ is an integral basis of $E_i/E_{i+1}$ at all places $Q\in\PP_{E_i}$ with $Q\neq P_{\infty}^{(i)}$.
\item[(P4)] There are $n\in\Z_+$ rational places $\{P_1,\ldots,P_n\}\in\PP_F$ that are splitting completely in $E/F$.
\end{enumerate}
\end{minipage}}

\begin{remark}
   The above assumptions (P2) and (P3) differ slightly with the assumptions in \cite{encAG24} for explicit construction of generators $y_1,y_2,\dots,y_r$ via automorphisms in $\Aut(E/\F_q(x))$. Since the generators $y_1,y_2,\dots,y_r$ will be used in our construction of a basis of the Riemann-Roch space $\cL(\lambda P^{(0)}_{\infty})$ of $E$ in the following Lemma~\ref{lem:basis}, we will give a proof of Lemma~\ref{lem:basis} in the Appendix.
\end{remark}

Let $k<n$ be a positive integer and $\lambda:=(k-1)m$. From assumptions (P1)-(P3), there is an explicit basis of the Riemann-Roch space $\cL(\lambda P^{(0)}_{\infty})$ of $E$. Before introducing the basis we need some notations. For a fixed factorization $m=\prod_{i=1}^rp_i$ of $m$, any $u\in[0,m-1]$ can be uniquely written as 
\begin{equation}\label{eq:expu}
u=u_0+u_1p_1+\ldots+u_{r-1}p_1p_2\cdots p_{r-1},\ \text{where}\ u_i\in[0,p_{i+1}-1],\ i=0,\dots,r-1.
\end{equation}
Let $\bu=(u_0,u_1,\ldots,u_{r-1})\in\prod_{i=1}^{r}\Z_{p_j}$ be the expansion coefficients of $u$ and define a function $\by^{\bu}$ as follows:
\begin{equation}\label{eq:expans}
\by^{\bu}:=y_0^{u_0}y_1^{u_1}\cdots y_{r-1}^{u_{r-1}}.
\end{equation}

\begin{lemma}{\cite{encAG24}}\label{lem:basis}
Suppose $E/F$ is a field extension satisfying (P1)-(P3). The set $\cB=\left\{\by^{\bu}x^j\mid \nu_{P^{(0)}_{\infty}}(\by^{\bu}x^j)\geq -\lambda,\ 0\leq u\leq m-1,\ 0\leq j\leq k-1\right\}$ is a basis of $\cL(\lambda P^{(0)}_{\infty})$. Therefore, every function $f\in\cL(\lambda P^{(0)}_{\infty})$ has a unique representation as
\[f=f_0(x)+f_1(x)\by^{\mathbf 1}+\dots+f_{m-1}(x)\by^{\mathbf{m-1}},\ \text{where}\ f_u(x)\in\F_q[x]_{<k}\ \text{for}\ u=0,\dots,m-1.\]
 \end{lemma}
\begin{proof}
See Appendix~\ref{app:basis}.
\end{proof}
\begin{remark}
    The above assumptions (P1)-(P3) are used to construct a polynomial basis of the Riemann-Roch space $\cL(\lambda{P}^{(0)}_{\infty})$. For some function field $E/\F_q$ that does not satisfy all assumptions (P1)-(P3), if its Riemann-Roch space has a good polynomial basis, then the AG codes based on $E$ can still be burst decoded in quasi-linear time by using same techniques. We give an example in Subection 5.5 to illustrate this point.
\end{remark}

From assumption (P4), assume $P\in\{P_1,\dots, P_n\}$ is a rational place that splits completely in $E/F$. Let $\fP_1,\dots,\fP_m$ be $m$ rational places in $E$ lying over $P$. Assume $f=\sum_{u=0}^{m-1}f_u(x)\by^{\mathbf u}\in\cL(\lambda P_{\infty}^{(0)})$. Note that each $f_u(x)\in\F_q[x]_{<k}$ is a constant function restricted on $\{\fP_1,\dots,\fP_m\}$, i.e., $f_u(\fP_1)=\dots=f_{u}(\fP_m)=f_u(P)\in\F_q$. Thus
\[\ev_{\{\fP_1,\dots,\fP_m\}}(f)=\ev_{\{\fP_1,\dots,\fP_m\}}\left(f_0(P)+f_1(P)\by^{\mathbf 1}+\dots+f_{m-1}(P)\by^{\mathbf{m-1}}\right),\]
For any $ f\in\cL(\lambda P^{(0)}_{\infty})$. let $f_P=\sum_{u=0}^{m-1}f_{u}(P)\by^{\mathbf{u}}$. We define the G-FFT and inverse G-FFT (G-IFFT for short) of $ f\in\cL(\lambda P^{(0)}_{\infty})$ at place $P$ as follows:
\begin{equation}\label{eq:fft}
\begin{split}&\GFFT_P:\ \left(f_0(P), f_1(P),\dots,f_{m-1}(P)\right) \rightarrow\left(f(\fP_1),\ldots,f(\fP_m)\right);\\
&\GIFFT_P:\ \left(f(\fP_1),\ldots,f(\fP_m)\right)\rightarrow \left(f_0(P), f_1(P),\dots,f_{m-1}(P)\right). \end{split}
\end{equation}
The local G-FFT and G-IFFT of length $m$ correspond to the column G-FFT and G-IFFT in the map from $\AG$ code to the interleaved $\RS$ code in the next subsection. The following lemma shows that local G-FFT and G-IFFT can be performed in time $O(m\log m)$. Its proof is similar to \cite[Theorem 3.3]{encAG24} which considers the whole G-FFT and G-IFFT of length $N=nm$. For the sake of completeness, we provide a proof in the Appendix.

\begin{lemma}{\cite{encAG24}}\label{lem:gfft}
Suppose $E/F$ is a field extension satisfying (P1)-(P4). Assume $m=\prod_{i=1}^r p_i$ and $B=\max\{p_1,\ldots,p_r\}$. Let $P$ and $\fP_1,\dots,\fP_m$ be defined as above. For any $f=\sum_{u=0}^{m-1} {\Ga_{u}}\by^{\bu}\in\F_q[y_0,y_1,\dots,y_{r-1}]$, the multipoint evaluation $\left(f(\fP_1),\ldots,f(\fP_m)\right)$ can be computed in time $O(B m\log m)$; conversely, given $\left\{f(\fP_1),\ldots,f(\fP_m)\right\}$, the coefficients  $\{\Ga_{0},\Ga_{1},\ldots,\Ga_{m-1}\}\subset\F_q^m$ of $f$ can be computed in time $O(Bm\log m)$.
\end{lemma}
\begin{proof}
See the Appendix~\ref{app:gfft}.
\end{proof}

Next, we present the \rAG codes from $E$. Assume $\{P_1,\dots,P_n\}$ as in (P4), let $\calP[j]$ be the set of all places in $E$ lying over $P_j$ and $\calP$ be the set of all places lying over $P_1,\dots,P_n$, i.e.,
\begin{equation}\label{eq:points}
\calP=\bigcup_{j=1}^n\calP[j]\ \text{and}\ \calP[j]=\{\fP_{1,j},\dots,\fP_{m,j}:\ \fP_{i,j}\mid P_j\}.
\end{equation}
Then $|\calP|:=N=nm$. From now on, we always assume the parameter $\lambda$ satisfying $2\fg(E)-1\leq \lambda <N$, where $\fg(E)$ is the genus of $E$. The one-point \rAG code is defined as
\begin{equation}\label{eq:agc}
\mathfrak{C}(\lambda P^{(0)}_{\infty},\calP)=\{\ev_{\calP}(f)=(\ev_{\calP[1]}(f),\ldots,\ev_{\calP[n]}(f))\mid f\in\cL(\lambda P^{(0)}_{\infty})\}.
\end{equation}

From the above Lemma~\ref{lem:gfft}, then $\mathfrak{C}(\lambda P^{(0)}_{\infty},\calP)$ can be encoded in time $O(N\log N)$ under the capability to perform FFT for any polynomial in $\F_q[x]_{<n}$ at $\{P_1,P_2,\dots,P_n\}$: For any $f=\sum_{u=0}^{m-1}f_u(x)\by^{\bu}\in\cL(\lambda P^{(0)}_{\infty})$ that is represented under the basis $\cB$, one can 
\begin{itemize}
\item[-]\textbf{Step 1}: compute $\{f_u(P_1),\ldots f_u(P_n)\}_{u=0}^{m-1}$ in time $mO(n\log n)$ via FFT;
\item[-]\textbf{Step 2}: compute $\left\{\ev_{\calP[j]}\left(\sum_{u=0}^{m-1}f_u(P_j)\by^{\bu}\right)\right\}_{j=1}^n$ in time $nO(m\log m)$ by Lemma~\ref{lem:gfft}. 
\end{itemize}
Hence, the total complexity to compute $\ev_{\calP}(f)$ is $O(N\log N)$.

\subsection{AG-based interleaved RS codes}\label{sec:AG&IRS}
For \rAG codes that are constructed from function field $E$ satisfying (P1)-(P4), we show that they can be one-to-one mapped to interleaved RS codes. Let $\{P_1,P_2,\dots,P_n\}$ and $\calP$ be defined as Equation~\eqref{eq:points}. We have the following AG-based $\IRS$ code from $E$
\[\footnotesize
\IRS_m(\lambda P^{(0)}_{\infty},\calP)=\left\{\sC_f=\left(\begin{matrix}
f_0(P_1)           &          f_0(P_2)    & \cdots  & f_0(P_{n})\\ 
f_1(P_1)           &          f_1(P_2)    & \cdots  & f_1(P_{n})\\ 
\vdots       & \vdots       & \cdots   & \vdots \\
f_{m-1}(P_1)           &          f_{m-1}(P_2)    & \cdots  & f_{m-1}(P_{n})\\  \end{matrix}\right)\mid f=\sum_{u=0}^{m-1}f_u(x)\by^{\bu}\in\cL(\lambda P^{(0)}_{\infty}) \right\}.\]
Note that $f_0(x),f_1(x),\dots,f_{m-1}(x)\in \F_q[x]_{<k}$ and $k\leq n$. Thus each row $(f_u(P_1),\dots,f_u(P_{n}))$ is a codeword in $\RS[n,k]$ and ${\IRS}_m(\lambda P^{(0)}_{\infty},\calP)$ is indeed an interleaved $\RS$ code. 

\begin{remark}\label{rmk:irs}
Let notations be defined as above. For $0\leq u\leq m-1$, let $k_u=\left\lfloor \frac{\lambda+\nu_{P^{(0)}_{\infty}}(\by^{\bu})}{m} \right\rfloor+1$. For any  $f=\sum_{u=0}^{m-1}f_u(x)\by^{\bu}$, {as} $\nu_{P_{\infty}^{(0)}}(x)=-m$, then
\[\begin{split}
    f\in\cL(\lambda P^{(0)}_{\infty}) &\Longleftrightarrow -\deg_x(f_u(x))m+\nu_{P_{\infty}^{(0)}}(\by^{\bu})\geq -\lambda\ \text{for\ all}\ u\\
    &\Longleftrightarrow \deg_x(f_u(x))<k_u\ \text{for\ all}\ u.
\end{split}\]
Thus ${\IRS}_m(\lambda P^{(0)}_{\infty},\calP)=\IRS_m\left(n,k_0,\dots,k_{m-1}\right)$ is a heterogenous $\IRS$ code defined on multipoint set $\{P_1,P_2,\dots,P_n\}$. We will use this fact in the probabilistic unique decoding of \rAG codes with burst errors.
\end{remark}

Define a map between $\fC(\lambda P^{(0)}_{\infty},\calP)$ and ${\IRS}_m(\lambda P^{(0)}_{\infty},\calP)$ as follows:
\[\begin{split}
\tau:\ &\fC(\lambda P^{(0)}_{\infty},\calP)\longrightarrow {\IRS}_m(\lambda P^{(0)}_{\infty},\calP),\ \bc=(\ev_{\calP[1]}(f),\ldots,\ev_{\calP[n]}(f))\mapsto\sC_f.\end{split}\]

\begin{theorem}\label{thm:convert}
Let $E/F$ be an algebraic function field satisfying (P1)-(P4). Then the map $\tau$ defined as above is a bijection between {$\AG$} code $\fC(\lambda P^{(0)}_{\infty},\calP)$ and the $\AG$-based $\IRS$ code ${\IRS}_m(\lambda P^{(0)}_{\infty},\calP)$. 
\end{theorem}
\begin{proof}
If there exists two functions $f=\sum_{u=0}^{m-1}f_u(x)\by^{\mathbf u}$ and $\tilde{f}=\sum_{u=0}^{m-1}\tilde{f}_u(x)\by^{\mathbf u}$ in $\cL(\lambda P^{(0)}_{\infty})$ such that $\sC_f=\sC_{\tilde{f}}$. Then for each $0\leq u\leq m-1$, $f_u(x)$ and $\tilde{f}_u(x)$ has $n$ same evaluations at $\{P_1,\dots,P_n\}$. As $\deg(f_u(x)),\deg(\tilde{f}_u(x))<k\leq n$, then $f_u(x)=\tilde{f}_u(x)$ hence $f=f'$. 

For any $\sC\in{\IRS}_m(\lambda P^{(0)}_{\infty},\calP)$, there exists an $f(x)=\sum_{u=0}^{m-1}f_u(x)\by^{\bu}\in\cL(\lambda P^{(0)}_{\infty})$ such that the $u+1$-th row of $\sC$ is $\sC(u+1)=\ev_{\{P_1,\dots,P_n\}}(f_u(x))$ for $0\leq u\leq m-1$. Then $\bc=\ev_{\calP}(f)\in \fC(N,\lambda)$ is the unique preimage of $\sC$ in $\fC(\lambda P^{(0)}_{\infty},\calP)$. 
\end{proof}

Actually, the transformation between $\fC(\lambda P^{(0)}_{\infty},\calP)$ and ${\IRS}_m(\lambda P^{(0)}_{\infty},\calP)$ is performed by the local G-FFT and G-IFFT at all $P_i$: for any $f\in \cL(\lambda P^{(0)}_{\infty})$,
\[
\small
\begin{split}
(\ev_{\calP[1]}(f),\ev_{\calP[2]}(f),\ldots,\ev_{\calP[n]}(&f))\xrightleftharpoons[m-unfolding]{m-folding}\left(\begin{matrix}
f(\fP_{1,1})           &         f(\fP_{1,2})      & \cdots  & f(\fP_{1,n})  \\ 
f(\fP_{2,1})            &         f(\fP_{2,2})      & \cdots  & f(\fP_{2,n})  \\ 
\vdots       & \vdots       & \cdots   & \vdots \\
f(\fP_{m,1})          &         f(\fP_{m,2})     & \cdots  & f(\fP_{m,n})  \\  \end{matrix}\right) \\
&\xrightleftharpoons[column\ \GFFT]{column\ \GIFFT}\left(\begin{matrix}
f_0(P_1)           &          f_0(P_2)    & \cdots  & f_0(P_{n})\\ 
f_1(P_1)           &          f_1(P_2)    & \cdots  & f_1(P_{n})\\ 
\vdots       & \vdots       & \cdots   & \vdots \\
f_{m-1}(P_1)           &          f_{m-1}(P_2)    & \cdots  & f_{m-1}(P_{n})\\  \end{matrix}\right).
\end{split}\]



\section{Decoding RS codes with burst errors up to Singleton bound}\label{sec:RS}	

In this section, we first show that the rational function field $\F_q(x)$ has a subfield tower satisfying (P1)-(P4) in Section~3. Thus, as a special kind of \rAG codes, $\RS[n,k]$ can be converted to $\IRS_m\left(n/m,\lceil k/m\rceil\right)$ for any $m\mid n$. Then we can use the decoding capacity of short $\IRS$ codes to analyze the decoding capacity of $\RS$ codes in the burst error model. We consider both list-decoding and probabilistic unique decoding of $\RS$ codes. The derived results are reflected in two aspects: 1) the burst-error decoding radius $\rho$ can reach the Singleton bound, namely $\rho\leq 1-\cR-\ep$, where $\cR=k/n$ is the code rate and $\ep>0$ is any small real number; 2) the decoding algorithms are quasi-linear $O_{\ep}(n\log n)$ regarding the code length $n$. 

Let us first recap the G-FFT framewrok in \cite{LX24}. In our burst decoding of $\RS$ code, we need to unify the multiplicative and additive subgroups of $\Aut(F/\F_q)$. Let $F=\F_q(x)$ and $\Aut(F/\F_q)\cong\PGL_2(\F_q)$ be the automorphism group of $F$. Assume $\F_q=\F_{\ell^{u}}$, where $\ell$ is a prime power. Let $T\leq \F_{\ell}^*$ be a multiplicative subgroup of order $t=\prod_{j=1}^{\mu}q_j$ and $W\leq \F_{\ell^{u}}$ be a subspace with dimension ${w}$. Consider the semidirect product $G:=T\ltimes W$, which is an affine linear subgroup of $\Aut(F/\F_q)$ of order $n=t\ell^{w}\mid q(q-1)$ and $n\leq q$. Let $$G_i= \{1\}\ltimes W_i,\ \text{for}\ 0\leq i\leq w,$$ where $W_{i}\leq W$ is a subspace of $W$ of dimension $i$, and $$G_{i+w}=T_{i}\ltimes W,\ \text{for}\ 1\leq i\leq \mu,$$ where $T_{i}\leq T$ is a cyclic subgroup of order $\prod_{j=1}^i q_j$. Then $G_{i-1}\leq G_i$ with index $[G_{i}: G_{i-1}]=\ell$ if $1\leq i\leq w$, and $[G_{i}: G_{i-1}]=q_i$ if $w+1\leq i\leq w+\mu$. Thus $G$ has a normal subgroup chain, i.e.,
\[\{0\}=G_0\trianglelefteq G_1\trianglelefteq\ldots \trianglelefteq G_{\mu+w}=G.\]
For the sake of brevity, we set $r=w+\mu$, $p_1=\dots=p_w=\ell$ and $p_{w+i}=q_i$ for $1\leq i\leq t$. Then $n=|G|=\prod_{i=1}^r p_i$. Let  $F_i=F^{G_i}$ be the subfield of $F$ fixed by $G_i$. By the Galois theory, the above subgroup chain of $G$ derives a subfield tower of $F$
\[F=F_0\supsetneq F_1\supsetneq\dots \supsetneq F_r,\ \text{and}\ [F_{i-1}: F_i]=p_i.\]
Moreover, $F_i=\F_q(x_i)$, where $x_i=\prod_{\sigma\in G_i}\sigma(x)$ is a polynomial of $x$ of degree $|G_i|$. By the property of the rational function field \cite{LX24}, the infinity place $P_{\infty}\cap F_r$ is totally ramified in $F/F_r$. Furthermore, let $P^{(i)}_{\infty}=P_{\infty}\cap F_i$. Then $\nu_{P^{(i)}_{\infty}}(x_i)=1$. Thus $\gcd(\nu_{P^{(i)}_{\infty}}(x_i),p_i)=1$. By Lemma~\ref{lem:basis} or \cite[Lemma 3.1]{LX24}, the Riemann-Roch space $\cL((n-1)P_{\infty})$ has a basis 
\begin{equation}\label{eq:basisx}
    \{\bx^{\bu}:=x_0^{u_0}x_1^{u_1}\cdots x_{r-1}^{u_{r-1}}\mid 0\leq u\leq n-1\}.
\end{equation}
Consequently, the field extension of rational function field $F/F^G$ satisfies (P1)-(P3) in \S~3. The following lemma presents the rational places in $F^G$ that split completely in $F$.

\begin{lemma}\label{lem:MPset}
Assume $\F_q=\F_{\ell^u}$. Let $F=\F_q(x)$ and $G=T\ltimes W$ be an affine linear subgroup of $\Aut(F/\F_q)$ of order $n=t\ell^{w}\leq q$ defined as above. Let $F^G=\F_q(x_r)$ and $L(x)=\prod_{\alpha\in W}(x+\alpha)$ be the linearized polynomial of $W$. Then the zero rational places of $\{x_r-L(\gamma)^t\mid \gamma\in\F_{\ell^u}\ L(\gamma)\neq 0\}$ in $\F_q(x_r)$ are splitting completely in $F$. Specifically, all places in $F$ lying over $x_r-L(\gamma)^t$ are $\left\{P_{\alpha}\in\PP_F\mid \alpha=u\cdot(\gamma+v)\mid u\in T, v\in W\right\}$.
\end{lemma}
\begin{proof}
All elements $\sigma$ in $G=T\ltimes W$ and $\sigma(x)$ are given by
\[\begin{split}&\sigma=\left(\begin{matrix} \beta & \alpha\\ 0 & \beta\end{matrix}\right),\ \text{where}\ \alpha\in W ,\ \beta\in T,\\
&\sigma(x):=\frac{\beta x+\alpha}{\beta}=\beta^{-1}(\beta x+\alpha).\end{split}\]
The rational function field $F^G=\F_q(x_r)$ and 
\[x_r=\prod_{\sigma\in G}\sigma(x)=\prod_{\beta\in T}\beta^{-1}\left(\prod_{\alpha\in W}(\beta x+\alpha)\right)=\prod_{\beta\in T}\beta^{-1}L(\beta x)=\prod_{\beta\in T}L(x)=L(x)^t.\]
For any $\gamma\in\F_{\ell^u}$ and $L(\gamma)\neq 0$, the equation $x_r=L(x)^t= L(\gamma)^t$ has $n$ solutions in $\F_{\ell^u}$, i.e., $\left\{x=u\cdot(\gamma+v)\mid u\in T, v\in W\right\}$. Hence, for each $L(\gamma)^t\neq 0$, $x_r-L(\gamma)^t$ splits completely in $F/F_r$.
\end{proof}

\begin{remark}\label{rmk:cyc}
    Let $\calP=\left\{\alpha=u\cdot(\gamma+v)\mid u\in T, v\in W\right\}$ be the multipoint set of size $n$ defined in Lemma~\ref{lem:MPset}. If $n=t\mid q-1$, then $W=\{0\}$ and $L(x)=x$. In this case, $\calP$ forms a cyclic subgroup of $\F_q^*$ or a coset of a cyclic subgroup. If $n=\ell^w t$ and $w\geq 1$, then $\calP$ is neither a cyclic subgroup nor a coset of a cyclic subgroup.
\end{remark}

Under the framework of \rAG codes in Section 3, we define the RS code that can be fast encoded and decoded via G-FFT. Let $k<n$ be an integer. We choose a rational place $\fp=(x_r-L(\gamma)^t)\in\PP_{F_r}$ that is splitting completely in $F/F_r$ by Lemma~\ref{lem:MPset}. Let $\calP=\{P_1,\ldots,P_n\}\subset \PP_F$ be the set of all rational places lying over $\fp$. Then we have an $\RS[n,k]$ defined as
\begin{equation}\label{eq:RS}
\RS[n,k]=\{\ev_{\calP}(f)\mid f=\sum_{u=0}^{k-1}\alpha_i\bx^{\bu}\in\cL((k-1)P_{\infty}\}.
\end{equation}

Next, using the G-FFT technique, we consider the AG-based $\IRS$ code derived from $\RS[n,k]$. For any $1\leq s\leq r$, let $m_s=\prod_{j=1}^sp_j$ be a factor of $n$, $n_s=n/m_s=\prod_{j=s+1}^rp_j$, and $k(s)=\lceil \frac{k}{m_s}\rceil$. Let $G_s\leq G$ be a subgroup of order $m_s$ and $F_s=F^{G_s}=\F_q(x_s)$ be the fixed subfield by $G_s$. Then $[F: F_s]=m_s$. We fix the following notations:
\begin{itemize}
\item Let ${\sP}_s=\{\fp_1,\dots,\fp_{n_s}\}=\calP\cap F_s$ be $n_s$ places in $F_s$ lying over $\fp$.
\item For each $j=1,2,\dots,n_s$, let $\calP_s[j]=\{P\in \calP\mid P\mid \fp_j\}$. Then $|\calP_s[j]|=m_s$.
\end{itemize}
We assume the places in the set $\calP$ are arranged in the following order: \[\calP=\calP_s[1]\cup\calP_s[2]\cup\dots\cup\calP_s[n_s].\] 
Consider the medium field extension $F/\F_q(x_s)$ of $F/F^G$, which also satisfies (P1)-(P3). The places in $\sP_s$ are splitting completely in $F$. Thus, by Lemma~\ref{thm:convert}, $\RS[n,k]$ can be mapped to the AG-based $\IRS$ code

\[
{\IRS}_{m_s}((k-1)P^{(0)}_{\infty},\calP)=\left\{\left(\begin{matrix}
\ev_{\sP_s}(f_0(x_s))  \\ 
\ev_{\sP_s}(f_1(x_s))  \\ 
\vdots       \\
\ev_{\sP_s}(f_{m_s-1}(x_s))  \\  \end{matrix}\right)\mid f=\sum_{u=0}^{m_s-1}f_u(x_s)\bx^{\bu}\in\cL((k-1)P^{(0)}_{\infty}) \right\}.\]
By Remark~\ref{rmk:irs}, $\IRS_{m_s}((k-1)P_{\infty},\calP)=\IRS_{m_s}(n,k_1,\dots,k_{m_s})$ is a heterogeneous $\IRS$ code, where
\[k_i=\left\lfloor\frac{k-1-i}{m_s} \right\rfloor+1,\ \text{for}\ i=0,1,\dots,m_s-1.\]
Since $k_{m_s-1}\leq \dots \leq  k_1\leq k_0=k(s)\leq n_s$, ${\IRS}_{m_s}((k-1)P^{(0)}_{\infty},\calP)\subset {\IRS}_{m_s}(n_s,k(s))$.

\subsection{Burst list decoding $\RS[n,k]_q$ in time $O_{\ep}(n\log n)$}
\begin{lemma}\label{lem:ldIRS}
The interleaved \rRS code $\IRS_{m}(n,k)$ is $(1-\cR,O(n))$-burst list-decodable, where $\cR=k/n$. 
\end{lemma}
\begin{proof}
Let $\fl=n-k$. For any $\sR\in\F_q^{m\times n}$, we consider the number of codewords $\sC\in\IRS_m(n,k)$ such that $\sC=\sR-\sE$ for some $\fl$-burst $\sE$ in $\F_q^{m\times n}$. 
As a burst $\sE$ corrupts all $m$ constituent $\RS$ codewords with the same locations, the bursts in all rows $\sR(1),\dots,\sR(m)$ of $\sR$ are in the same burst interval. Note that the erasure decoding radius of $\RS[n,k]$ is $n-k$. Since for any codeword $\br\in\F_q^n$ with erased positions $\leq n-k$, then reminder $k$ entries of $\br$ uniquely determine a codeword in $\RS[n,k]$ by interpolation. Thus, for a burst interval $I\subset[1,n]$ of length $n-k$, by erasing the columns of $\sR$ with indices in $I$ and then decoding each row of $\sR$ with the left entries, we can obtain at most one codeword in $\IRS_m(n,k)$. There are $k+1$ burst intervals of length $n-k$ in total. Thus, at most $k+1=O(n)$ codewords can be decoded from $\sR$.  Consequently, we have that $\IRS_m(n,k)$ is $(1-\cR, O(n))$-list decodable.
\end{proof}

 \begin{algorithm}[!h]\label{alg: ListDec}
  \caption{Burst list decoding($\br$).}
  \label{alg:ListDec}
  \begin{algorithmic}[1]
    \Require
     $\br\in\F_q^n$.
    \Ensure
      $\mathfrak L=\left\{f\in\F_q[x]_{<k}\mid \ev_{\calP}(f)-\br\ \text{is\ an}\ \fl\text{-burst}\right\}$ for $\fl=n-k-2\prod_{j=1}^{s}p_j$.
    \State $\mathfrak L=\emptyset$.
    \State Fold $\br$ with length $m_s=\prod_{j=1}^{s}p_j$ into $\br=(\br^T_1,\br^T_2,\ldots,\br^T_{n_s})\in\F_q^{{m_s}\times{n_s}}$.
    \State For $j=1,\dots,n_s$, use G-IFFT to recover $g_j(x)=g_{0,j}+g_{1,j}\bx^{\mathbf 1}+\ldots+g_{m_s-1,j}\bx^{\mathbf{m_s-1}}$.
    \For {every burst interval $I\subset[n_s]$ of length $n_s-k(s)$}
           \For {$i=0$ to $m_s-1$}
                  \State Recover $f_i(x_s)\in\F_q[x_s]_{<k(s)}$ from $\{g_{i,j} \mid j\in [n_s]\setminus I \}$ via erasure decoding.
            \EndFor
            \State Compute $f(x)=f_0(x_s)+f_1(x_s)\bx^{\mathbf 1}+\ldots f_{m_s-1}(x_s)\bx^{\mathbf{m_s-1}}$.
            \State If $\deg(f)<k$ and $\br-\ev_{\calP}(f)$ is an $\fl$-burst, $\mathfrak L=\mathfrak L\cup\{f(x)\}$.
     \EndFor
     
  \State  \Return  $\mathfrak L$.
  \end{algorithmic}
\end{algorithm}

\begin{theorem}\label{thm:listDecRS}
Let $\RS[n,k]$ be defined as in Equation~\eqref{eq:RS}. For any real number $\epsilon >0$, the $\RS$ code $\RS[n,k]$ is $(1-\cR-\epsilon,O(1/\epsilon))$-burst list-decodable, where $\cR=k/n$ is the code rate. Furthermore, if $n$ is an $O(1)$-smooth integer, then $\RS[n,k]$ can be list decoded in time $O(\frac 1{\epsilon} n\log n)=O_{\ep}(n\log n)$ by Algorithm 1.
\end{theorem}
\begin{proof}
Given an erroneous codeword $\br$ corrupted by a burst error of length at most $\fl:=n-k-n\epsilon$, we will show that 
\[|B_{\fb}(\br, \fl)\cap \RS[n,k]|=O(1/\ep).\]

Let $m_s=\prod_{j=1}^sp_j\mid n$ and $n_s=n/m_s$, $k(s)=\lceil k/m_s\rceil$. Assume $\bc\in B_{\fb}(\br, \fl)\cap \RS[n,k]$ and $\br=\bc+\be$ for some burst $\be$ of length $\fl$. 
By folding $\br$, $\bc$ and $\be$ with length $m_s$, respectively, we have
\[(\br^T_1,\ldots,\br^T_{n_s})=(\bc^T_1+\be^T_1,\dots,\bc^T_{n_s}+\be^T_{n_s}).\]
By Lemma~\ref{lem:easy}, the number of nonzero $\be_j\neq 0$ is at most $\lfloor \fl/m_s\rfloor+2$. 
By performing G-IFFT on each columns of $(\br^T_1,\ldots,\br^T_{n_s})$, $(\bc^T_1,\dots,\bc^T_{n_s})$ and $(\be^T_1,\dots,\be^T_{n_s})$, we obtain matrices $\sR,\sC,\sE\in \F_q^{m_s\times n_s}$ such that
\[\mathsf R=\left(\begin{matrix}
g_{0,1}            &           g_{0,2}    & \cdots  & g_{0,n_s}\\ 
g_{1,1} & g_{1,2} & \cdots  & g_{1,n_s} \\
\vdots       & \vdots       & \cdots   & \vdots \\
g_{m_s-1,1} &   g_{m_s-1,2} & \cdots &   g_{m_s-1,n_s} \end{matrix}\right)=\sC+\sE,\]
where the $j$-th column $\sR[j]$ are the coefficients of a function $g_j(x):=\sum_{i=0}^{m_s-1}g_{i,j}\bx^{\mathbf i}\in\F_q[x]_{<m_s}$ satisfying $\ev_{\calP[j]}(g_j)=\br_j$. Note that when $\be_j=0$, then the $j$-th column $\sR[j]$ of $\sR$ is $\sC[j]$. If $\be_j\neq 0$, then $\sR[j]\neq \sC[j]$. Consequently, any $\bc\in B_{\fb}(\br, \fl)\cap \RS[n,k]$ corresponds to an $\IRS$ codeword $\sC$ in the ball centered at $\sR$ with radius $\leq \lfloor \fl/m_s\rfloor+2$. 

By Lemma~\ref{lem:ldIRS}, ${\IRS}_{m_s}((k-1)P^{(0)}_{\infty},\calP)$ is $O(1-\cR_s,O(n_s))$-burst list-decodable, where $\cR_s=k(s)/n_s$ and $k(s)=\lceil k/m_s\rceil$. Thus when the length $\fl$ of burst error satisfying $\lfloor \fl/m_s\rfloor+2\leq n_s-k(s)$, i.e., $\fl\leq n-k-2m_s$, $\RS[n,k]$ is burst list-decodable with list size at most $O(n_s)$. For any real number $0< \ep<1-\cR$, by setting $n_s=2/\ep$, we have that $\RS[n,k]_q$ is $O(1-\cR-2/n_s,O(n_s))=(1-\cR-\ep, O(1/\ep))$-burst-list-decodable.

The above proof shows that burst-list decoding $\RS$ code is determined by burst-list decoding $\IRS$ code. Algorithm 1 is designed according to the proof of Lemma \ref{lem:ldIRS}. Let $\sR$ be the erroneous $\IRS$ codeword obtaining by running G-IFFT of $(\br^T_1,\dots,\br^T_{n_s})$. For each burst interval $I\subset [n_s]$ of length $n_s-k(s)$, we can decode all rows of $\mathsf R$ with erasure locations in $I$ simultaneously, obtaining at most one codeword $\sC\in{\IRS}_m((k-1)P^{(0)}_{\infty},\calP)$. Assume
\[\sC=\left(\begin{matrix}
\ev_{\sP_s}(f_0(x_s))\\
\vdots\\
\ev_{\sP_s}(f_{m_s-1}(x_s))\end{matrix}\right),\ f_0(x_s),\dots,f_{m_s-1}(x_s)\in \F_q[x_s]_{<k(s)}.\] 
Let $f(x)=\sum_{i=0}^{m_s-1}f_i(x_s)\bx^{\bi}$. If $\deg_x(f(x))<k$, then $\bc=\ev_{\calP}(f)\in\RS[n,k]_q$ is a codeword in the burst error ball centered at $\br$ with radius $n-k$. We only put the codeword $\bc\in B_{\br}(\br,n-k-2m_s)$ into our decoding list. There are at most $k(s)+1$ burst intervals of length $n_s-k(s)$. Thus we can decode at most $k(s)+1$ codewords $\bc$ from $\br$.

For computational complexity analysis, we see that Step 3 costs $O(n_sBm_s\log m_s)$ operations in $\F_q$ via G-IFFT by Lemma~\ref{lem:gfft}, where $B=\max\{p_i\mid i\in[r]\}$; Steps 4-7 cost at most $m_s\cdot k(s)n_s\log n_s=O(|\mathfrak L|n\log n_s)$ via fast erasure decoding algorithms for RS code~\cite{lin16novel} (see also the Step 5 (i)-(iii) in Algorithm 3 in the Appendix). Step 8 costs at most $O(n)$ operations. Hence, the total decoding complexity is $O(n_sB n\log m_s)+|\mathfrak L|\cdot O(n\log n_s)+O(n)=O(L n\log n)=O_{\ep}(n\log n)$ if $B=O(1)$.
\end{proof}

\subsection{Unique decoding RS codes in time $O_{\ep}(n\log n)$}
Let $\RS[n,k]$ be defined as in Equation~\eqref {eq:RS} with certain multipoint set $\calP$. When $n=t\mid q-1$, by Lemma~\ref{lem:MPset}, $\calP$ is a cyclic subgroup of order $n$ or a coset of a cyclic group. In this case, $\RS[n,k]$ can correct burst errors of length less than $n-k-e$ with miscorrection probability at most $1/q^e$ by Lemma~\ref{lem:UniDec}. When $n=\ell^wt\nmid q-1$ and $t>1$, one can not use Wu's algorithm directly for $\RS[n,k]$. In this case, we choose $n_s\mid t$ hence $n_s\mid\ell-1$. By Remark~\ref{rmk:cyc}, the multipoint set $\sP_s$ in $\IRS_{m_s}((k-1)P^{(0)}_{\infty},\calP)$ forms a cyclic subgroup of order $n_s$ or a coset of it. Then we convert the decoding problem in $\RS[n,k]$ to the decoding problem in $\IRS_{m_s}(n_s,k(s))$ which is solvable by Wu's algorithm. Finally, we can get the correct codeword in $\RS[n,k]$ by the inverse map.


\begin{lemma}\label{lem:UniIRS}
Assume $\alpha\in  \F_q^*$ is an element of order $n$. Let $\sP=\langle \alpha\rangle$ or a coset in $\F_q^*/\langle \alpha\rangle$. Let $\IRS_m(n,k_1,\ldots,k_m)$ be a heterogeneous $\IRS$ code defined on $\sP$ and $k=\max\{k_i:\ i\in[m]\}$. Let $e\geq 1$ be an integer.
Then $\IRS_m(n,k_1,\ldots,k_m)$ can correct burst error of length less than $n-k-e$ with miscorrection probability at most $m/q^e$.
\end{lemma}
\begin{proof}
Assume $\sR=\sC+\sE$ is a received erroneous $\IRS$ codeword, where $\sE$ is an $\fl$-burst in $\F_q^{m\times n}$ with length $\fl< n-k-e$. For each row $\sR(i)$ of $\sR$, since $\fl<n-k-e\leq n-k_i-e$, by Lemma~\ref{lem:UniDec},  we can find the correct $\sC(i)$ with miscorrection probability at most $1/q^e$. Thus, after running Wu's algorithm for all rows $\sR(1),\dots,\sR(m)$, $\sC$ can be correctly computed with miscorrection probability at most $m/q^e$.
\end{proof}

 \begin{algorithm}[!h]\label{alg: UniqDec}
  \caption{Probabilistic unique decoding($\br$).}
  \label{alg:UniqDec}
  \begin{algorithmic}[1]
    \Require
     $\br\in\F_q^n$ with burst error of length $\fl< n-k-n\epsilon(e+2)$.
    \Ensure
      $\bc\in  \RS[n,k]\cap B_{\fb}(\br,\fl)$.
    \State Fold $\br$ with length $m_s=\prod_{j=1}^{s}p_j$ into $\br=(\br^T_1,\br^T_2,\ldots,\br^T_{n_s})\in\F_q^{m_s\times n_s}$.
     \For {$j=1$ to $n_s$}
            \State Run G-IFFT on $\br_j$, obtaining a polynomial $g_j(x)=g_{0,j}+g_{1,j}\bx^{\mathbf 1}+\ldots+g_{m_s-1,j}\bx^{\mathbf{m_s-1}}$.
     \EndFor
     \For {$i=0$ to $m_s-1$}
           \State Run probabilistic unique decoding Algorithm 3 in Appendix A with input $\mathbf g_i=(g_{i,1},\dots,g_{i,n_s})$, obtaining $\sC(i)=(\sC_{i,1},\dots,\sC_{i,n_s})$.
           \State Compute a polynomial $f_i(x_s)\in\F_q[x_s]_{<k_i}$ from $\sC(i)$ by G-IFFT.
      \EndFor
     \State Compute $f(x)=f_0(x_s)+f_1(x_s)\bx^{\mathbf 1}+\ldots f_{m_s-1}(x_s)\bx^{\mathbf{m_s-1}}$.
     \State Compute $\bc=\ev_{\calP}(f(x))$.
     \State  \Return $\bc$.
  \end{algorithmic}
\end{algorithm}

\begin{theorem}\label{thm:uniqDecRS}
For any real number $\ep>0$ and an integer $e\geq 1$, $\RS[n,k]$ can probabilistic uniquely correct burst errors of length less than $n-k-n\ep(e+2)$ with miscorrection probability at most $\ep/q^{e-1}$.
Moreover, if $n$ is an $O(1)$-smooth integer, then probabilistic unique decoding can be performed by Algorithm 2 in time $O(\frac{1}{\ep}n\log n)$.
\end{theorem}
\begin{proof}
Assume $n=\ell^wt=\prod_{j=1}^rp_j$, where $t\mid \ell-1$. Choose $n_s=\prod_{j=s+1}^rp_j\mid t$ and $n_s=1/\ep$. Hence $n$ is $\ell$-smooth if $w\neq 0$, and $O(1/\ep)$-smooth for sufficiently small $\ep$ if $w=0$. Overall, $n$ is $O(1/\ep)$-smooth. Let $m_s=n/n_s=\prod_{j=1}^sp_j$. By Theorem~\ref{thm:convert}, we can map $\RS[n,k]$ to the AG-based $\IRS$ code $\IRS_{m_s}((k-1)P_{\infty},\calP)=\IRS_{m_s}(n,k_1,\dots,k_{m_s})$. In this case, by Lemma~\ref{lem:MPset}, the multipoint set $\sP_s$ used to define $\IRS_{m_s}((k-1)P^{(0)}_{\infty},\calP)$ is $\langle\xi\mid \xi^{n_s}=1\rangle \cdot L(\gamma)^{t/n_s}$, which is a coset in $\F_q^*/\langle \xi\rangle$ of order $n_s$. 

For any erroneous word $\br=\bc+\be\in\F_q^n$ corrupted by an $\fl$-burst $\be$, after folding and running G-IFFT, we get an erroneous word $\sR=\sC+\sE$ corrupted by an $\lfloor\fl/m_s\rfloor+2$-burst $\sE$ in $\F_{q}^{m_s\times n_s}$. 
By Lemma~\ref{lem:UniIRS}, $\IRS_{m_s}(n,k_1,\dots,k_{m_s})$ can correct  burst errors with length less than $n_s-k_0-e$ with miscorrection probability at most $m_s/q^e\leq \ep/q^{e-1}$ and any $\IRS$ codeword $\sC$ uniquely corresponds to an $\RS$ codeword $\bc=\tau^{-1}(\sC)$. Thus, $\RS[n,k]$ can correct bursts of length $\fl< n-k-n\ep(e+2)$ (which is derived from $\fl/m_s+2< n_s-k_0-e$), and the miscorrection probability of unique decoding is less than $\ep/q^{e-1}$.

From the above proof, the unique decoding problem in $\RS[n,k]$ is determined by the unique decoding problem in $\IRS_{m_s}(n,k_1,\dots,k_{m_s})$. According to the proof of Lemma~\ref{lem:UniIRS}, we give the probabilistic unique decoding algorithm as in Algorithm~2. The decoding cost includes the complexity of probabilistic unique decoding Algorithm~3 for $\RS$ code and the G-IFFT. Steps 2-4 cost $n_sO(\frac{1}{\ep}m_s\log m_s)$ operations;  Steps 5-8 cost $m_sO(n_s\log n_s)$ operations; Step 9 costs $O(n)$ operations and Step 10 costs $O(\frac{1}{\ep}n\log n)$ operations by G-FFT. Therefore, the total complexity is $O(\frac 1{\ep}n\log n)$.
\end{proof}

\section{Decoding \rAG codes with burst errors}\label{sec:AG}
In this section, we consider burst decoding of algebraic geometry codes $\mathfrak{C}(\lambda P^{(0)}_{\infty},\calP)$ that constructed in Section~\ref{sec:IRS}. As these AG codes can be mapped to interleaved \rRS code $\IRS_m[n,k]$, we will show that if the constituent $\RS[n,k]$ satisfies the construction given in Section~\ref{sec:RS}, then the list and probabilistic unique decoding results for $\RS[n,k]$ can be extended to $\mathfrak{C}(\lambda P^{(0)}_{\infty},\calP)$. 

\subsection{List decoding \rAG codes}

\begin{theorem}\label{thm:LDAG}
Let $E/F$ be an algebraic function filed satisfying (P1)-(P4) and the algebraic geometry code $\fC(\lambda P^{(0)}_{\infty},\calP)$ be defined as in Equation~\eqref{eq:agc}. Let $\IRS_{m}(\lambda P^{(0)}_{\infty},\calP)$ be the corresponding AG-based $\IRS$ code. Suppose that the constituent $\RS[n,k]$ of $\IRS_{m}(\lambda P^{(0)}_{\infty},\calP)$ satisfies assumptions in Theorem \ref{thm:listDecRS}. Then $\fC(\lambda P^{(0)}_{\infty},\calP)$ is $\left(1-\cR-\fg/N-\epsilon, O(1/\epsilon)\right)$-burst list-decodable, where $\cR$ is the code rate. Furthermore, if $n$ and $m$ are both $O(1)$-smooth, it can be list-decoded in time $O_{\ep}(N\log N)$. 
\end{theorem}
\begin{proof}
Since $\lambda$ satisfying $2\fg(E)-1\leq \lambda <N$, the code rate $\cR=(\lambda+1-\fg)/N$. Given an erroneous codeword $\br$ corrupted by a burst error of length at most $N(1-\cR-\fg/N-\ep)\leq \fl:=N-\lambda-N\epsilon$, we will show that 
\[|B_{\fb}(\br, \fl)\cap \fC(\lambda P^{(0)}_{\infty},\calP)|=O(1/\ep).\]

By assumptions (P1)-(P4), let $m=[E:F]=\prod_{i=1}^rp_i$. The field extension $E/\F_q(x)$ has a field tower
\begin{equation}\label{eq:towerE}
    E_0\supsetneq E_1\supsetneq \dots\supsetneq E_r=\F_q(x),\ [E_{i-1}: E_i]=p_i,\ i=1,\dots,r.
\end{equation}
If the the constituent $\RS[n,k]$ of $\IRS_{m}(\lambda P^{(0)}_{\infty},\calP)$ satisfies assumptions in Theorem \ref{thm:uniqDecRS}, i.e., $n$ has a factorization $n=\prod_{j=1}^{t}p_{r+j}$ and $\F_q(x)$ also has a tower of subfields: $$\F_q(x_{j-1})\supsetneq\F_q(x_{j})\ \text{with}\ [\F_q(x_{j-1}):\F_q(x_{j})]=p_{r+j}\ \text{for}\ j=1,\dots,t,$$ then the field tower in Equation~\eqref{eq:towerE} can be expanded to a field tower of length $r+t$:
\begin{equation}\label{eq:towerE}
    E_0\supsetneq \dots\supsetneq E_r=\F_q(x)=\F_q(x_0)\supsetneq \F_q(x_1)\supsetneq \dots\supsetneq \F_q(x_{t}).
\end{equation}
All assumptions (P1)-(P4) still hold for $E/\F_q(x_t)$ with extension degree $N=mn$. In this case, the basis of $\cL(\lambda P^{(0)}_{\infty}$ to perform G-FFT of length $nm$ in Lemma~\ref{lem:basis} is
\begin{equation}\label{eq:basislong}
   \cB_{r+t}:= \left\{\bx^{\bj}\by^{\bu}\mid \nu_{P^{(0)}_{\infty}}(\bx^{\bj}\by^{\bu})\geq -\lambda,\ 0\leq u\leq m-1,\ j\geq 0\right\},
\end{equation}
where $\bx^{\bj}$ and $\by^{\bu}$ are defined as Equations ~\eqref{eq:expans} and \eqref{eq:basisx}, respectively. Under the above basis, every function $f\in\cL(\lambda P^{(0)}_{\infty})$ has a unique representation as $f=\sum_{j,u}a_{j,u}\bx^{\bj}\by^{\bu},$ where $a_{j,u}\in\F_q$.

Similarly, for some $s\in[1,t]$, let $n_s=\prod_{j=s+1}^tp_{r+j}$ , $m_s=n/n_s$, and $k(s)=\lceil k/m_s\rceil$. Then we can map $\fC(\lambda P^{(0)}_{\infty},\calP)$ into an $m_sm$-interleaved $\RS$ code $\IRS_{m_sm}(n_s,k(s))$. The left discussion is the same as in the proof of Theorem~\ref{thm:listDecRS}. 
By folding $\br$ with length $m_sm$ and then perform G-IFFT on each column, any codeword $\bc$ in $B_{\fb}(\br, \fl)\cap \fC(\lambda P^{(0)}_{\infty},\calP)$ corresponds to a codeword $\sC$ in $B_{\fb}(\sR, \lfloor \fl/mm_s\rfloor+2)\cap \IRS_{mm_s}(\lambda P^{(0)}_{\infty},\calP)$. Since $\IRS_{m_sm}(n_s,k(s))$ is $(1-k(s)/n_s,O(n_s))$-burst list decodable, by taking $n_s=2/\ep$, $\fC(\lambda P^{(0)}_{\infty},\calP)$ can be list decoded with list size $O(1/\ep)$ when burst error of length $\fl$ satisfying
\[\fl\leq m_sm(n_s-k(s))-2N/n_s= N-\lambda-N\ep.\]


%

The decoding algorithm for $\mathfrak{C}(\lambda P^{(0)}_{\infty},\calP)$ is very similar to Algorithm 1, with code length $N$ and fold length $m_sm$. Hence, if $N=mn=\prod_{i=1}^rp_i\prod_{j}^tp_{r+j}$ is $O(1)$-smooth, by the local G-FFT and G-IFFT of length $m_sm$, the total decoding complexity is
$$n_sO(\frac{1}{\ep}m_sm\log m_sm)+mm_s{O}(n^2_s\log n_s)={O}(\frac{1}{\ep}N\log N)={O}_{\ep}(N\log N).$$

\end{proof}

\subsection{Unique decoding of \rAG codes}
Before introducing our probabilistic probabilistic unique decoding of \rAG codes for burst errors, let us derive the limit of deterministic unique decoding of AG codes for burst errors.

As every code can uniquely correct adversary errors up to half of the minimum distance with probability $1$, it also uniquely corrects burst errors up to half of the minimum distance  with probability $1$. Half of the minimum distance is also the limit for AG codes to correct burst errors.  We prove this result in the following lemma. 

\begin{lemma}\label{lem:uniqburst}
    Assume $\fC(\lambda P_{\infty},\calP)$ is an \rAG code of length $N$, dimension $K=\lambda+1-\fg$, where $\fg$ is the genus of the underlying curve.  If $\lambda\ge \fg$, then  $\fC(\lambda P_{\infty},\calP)$ can deterministically correct at most $\lfloor\frac{n-\lambda+\fg-1}2\rfloor$ burst errors.
\end{lemma}
\begin{proof} Suppose that it were false. Then $\fC(\lambda P_{\infty},\calP)$ can correct $1+\lfloor\frac{N-\lambda+\fg-1}2\rfloor$ burst errors. Let $\cP =\{P_1,\dots,P_N\}$ be the evaluation point set for  $\fC(\lambda P_{\infty},\calP)$. Set $t=1+\lfloor\frac{N-\lambda+\fg-1}2\rfloor$. Choose a nonzero function from the Riemann-Roch space $\cL(\lambda P_{\infty}-\sum_{i=1}^{\lambda-\fg}P_i)$ (note that this is possible as its dimension is at least $1$). Let $\be=(0,\dots,0,f(P_{\lambda-\fg+1}),\dots,f(P_{\lambda-\fg+t}), 0, \dots,0)=(f(P_1),\dots,f(P_{\lambda-\fg}),f(P_{\lambda-\fg+1}),\dots,f(P_{\lambda-\fg+t}), 0, \dots,0)$ be a vector in $\F_q^n$. Thus, the codeword $(f(P_1),\dots,f(P_N))$ can become $\be$ with at most $N-(\lambda-\fg+t)\le t$ burst errors, while the all-zero codeword can become $\be$ with at most $t$ burst errors. This means that the word $\be$ has burst Hamming distance at most $t$ from both the all-zero codeword and $(f(P_1),\dots,f(P_N))$. The desired result follows. 
\end{proof}

\begin{theorem}\label{thm:uniqDecAG}
Let $E/F$ be an algebraic function filed satisfying (P1)-(P4) and the algebraic geometry code $\fC(\lambda P^{(0)}_{\infty},\calP)$ be defined as in Equation~\eqref{eq:agc}. Let $\IRS_{m}(\lambda P^{(0)}_{\infty},\calP)$ be the corresponding AG-based $\IRS$ code. Suppose that the constituent $\RS[n,k]$ of $\IRS_{m}(\lambda P^{(0)}_{\infty},\calP)$ satisfies the assumptions in Theorem~\ref{thm:uniqDecRS}. Then for any real number $\ep>0$ and positive integer $e$ such that $\ep m\leq q^{e-1}$, $\fC(\lambda P^{(0)}_{\infty},\calP)$ can probabilistic uniquely correct burst errors of length less than $N\left(1-\cR-\fg/N-\ep(2+e)\right)$ with miscorrection probability at most $\ep m/q^{e-1}$, where $\cR$ is the code rate. Furthermore, if both $m$ and $n$ are $O(1)$-smooth integers, it can be uniquely decoded in time $O_{\ep}(N\log N)$.
\end{theorem}
\begin{proof}
For any $\ep>0$, let $n_s=1/\ep$, $m_s=n/n_s$ and $k(s)=\lceil k/m_s\rceil$ be defined as in the proof of Theorem~\ref{thm:LDAG}. Thus $N=n_sm_sm$ and $\lambda=(k-1)m=(k(s)-1)m_sm$. By assumptions, we can map $\fC(\lambda P^{(0)}_{\infty},\calP)$ into AG-based $m_sm$-interleaved $\RS$ code $\IRS_{m_sm}(n_s,k(s)\lambda P^{(0)}_{\infty},\calP)$. For probabilistic unique decoding, we need the dimension of each constituent $\RS$ code in $\IRS_{m_sm}(n_s,k(s)\lambda P^{(0)}_{\infty},\calP)$. Under the basis $\cB_{r+t}$ as in Equation~\eqref{eq:basislong}, any $f\in \cL(\lambda \fP_{\infty})$ can be decomposed as
\[f=\sum_{u=0}^{m-1}\sum_{j=0}^{m_s-1} f_{u,j}(x_s)\bx^{\bj}\by^{\bu},\ -m(m_s\deg(f_{j,u})+j)-\nu_{P^{(0)}_{\infty}}(\by^{\bu})\geq -\lambda. \]
Thus, $\IRS_{m_sm}(n_s,k(s)\lambda P^{(0)}_{\infty},\calP)={\IRS}_{mm_s}(n_s,\mathbf{k})$ defined on the multipoint set $\sP_s$, where
\[\begin{split}&\mathbf{k}=(k_{0,0},\dots,k_{0,m_s-1},\dots,k_{m-1,0},\dots,k_{m-1,m_s-1}),\ \text{and}\\
 k_u=\left\lfloor \frac{\lambda+\nu_{P^{(0)}_{\infty}}(\by^{\bu})}{m} \right\rfloor&+1,\ k_{u,j}=\left\lfloor \frac{k_u-1-j}{m_s} \right\rfloor+1,\ \text{for}\ 0\leq u\leq m-1,\ 0\leq j\leq m_s-1.\end{split}\]
Hence $k(s)=\max\{k_{u,j}\mid 0\leq u\leq m-1,\ 0\leq j\leq m_s-1\}$. 

Given any erroneous word $\br=\bc+\be\in\F_q^N$ corrupted by an $\fl$-burst $\be$, one can first fold it with length $mm_s$ and then perform G-IFFT on its $n_s$ columns, obtaining an erroneous $\IRS$ codeword $\sR=\sC+\sE\in\F_q^{mm_s\times n_s}$ which is corrupted by an $\lfloor \fl/mm_s\rfloor+2$-burst $\sE$. From assumption that $\fl<N\left(1-\cR-\fg/N-\ep(2+e)\right)\leq N-\lambda-N\ep(2+e)$ and $n_s=1/\ep$, we have $\lfloor \fl/mm_s\rfloor< n_s-k(s)-e-2$. By the probabilistic unique decoding results for interleaved $\RS$ code in Lemma~\ref{lem:UniIRS}, we can correct all rows $\sR(1),\dots,\sR(mm_s)$ of $\sR$ and get $\sC$ with miscorrection probability at most $mm_s/q^{e}\leq \ep m/{q^{e-1}}$.

The complexity analysis of the above probabilistic unique decoding of \rAG code is the same as in the decoding complexity analysis in Theorem~\ref{thm:uniqDecRS} with code length $N$ and fold length $m_sm$, which is $O(\frac{1}{\ep} N\log N)$.
\end{proof}


\subsection{Examples of {$\AG$} codes from Hermitian curve}
In this subsection, we give examples of  \rAG codes that satisfy all the assumptions in Theorem~\ref{thm:LDAG} and ~\ref{thm:uniqDecAG}.

\begin{example}{(Hermitian curve~ \cite{sti09})}\label{ex:hc}
Let $\kappa$ be a power of a prime $p$. Let $\F_q=\F_{\kappa^2}$ be a finite field. Consider the Hermitian curve $\H$ over $\F_q$ defined by
\[\H(x,y):=y^{\kappa}+y-x^{\kappa+1}.\]
Then $\H(x,y)$ is absolutely irreducible. Let $E=\F_q(x, y)=\mathrm{Frac}(\F_q[x,y]/\H(x,y))$ be the function field of $\H$. On the one hand, $E$ can be viewed as a Kummer extension over $\F_q(y)$ by extending $\F_q(y)$ with a variable $x$ satisfying the relation $\H(x,y)=0$. On the other hand, $E$ can also be viewed as an Artin-Schreier extension over $\F_q(x)$ by extending $\F_q(x)$ with a variable $y$ satisfying the relation $\H(x,y)=0$. The following Lemma ensures that $E$ satisfies the assumption (P1)-(P4) in Section 3. 

\begin{lemma}{\cite{sti09}}\label{lem:preH}
Let $E$ be the Hermitian function field.
\begin{enumerate}
\item The pole place ${P}_{\infty}$ of $x$ in ${\F_q(x)}$ (resp. $Q_{\infty}$ of $y$ in ${\F_q(y)}$) is totally ramified in $E$. Let ${\fP}_{\infty}$ be the unique place in $E$ lying above ${P}_{\infty}$ and $Q_{\infty}$. Then the discrete valuations of $x, y$ at ${\fP}_{\infty}$ are
\[\nu_{{P}_{\infty}}(x)=-\kappa,\ \nu_{{P}_{\infty}}(y)=-(\kappa+1).\]
Moreover, for positive integer $\lambda$, the Riemann-Roch space $\cL(\lambda \fP_{\infty})$ has the following two bases 
\[\begin{split}
&\cB_1=\{x^jy^u\mid j\kappa+u(\kappa+1)\leq \lambda,\ 0\leq u\leq \kappa-1,\ j\geq 0\},\ \text{and}\\ 
&\cB_2=\{y^jx^v\mid j(\kappa+1)+v\kappa\leq \lambda,\ 0\leq v\leq \kappa,\ j\geq 0\}.\end{split}\] 
\item The filed extension $E/\F_q(x)$ is Galois of degree $\kappa=[E: \F_q(x)]$, and its Galois group $\Gal(E/\F_q(x))=\{\sigma \mid \sigma(y)= y+\alpha,\ \alpha\in \F_{\kappa}\}$ is isomorphic to $\F_{\kappa}$. Moreover, all rational places in $\PP_{\F_q(x)}$ except $P_{\infty}$ are splitting completely in $E/\F_q(x)$.
\item The extension $E/\F_q(y)$ is Galois of degree $\kappa+1=[E: \F_q(y)]$ and its Galois group $\Gal(E/\F_q(y))=\{\sigma \mid \sigma(x)=\omega x,\ \omega^{\kappa+1}=1\}$ is isomorphic to the $\kappa+1$-cyclic subgroup of $\F_q^*$. Moreover, the set of rational places $\{P_{\beta}\in\PP_{\F_q(y)}\mid \beta^{\kappa}+\beta\neq 0\}$ are splitting completely in $E/\F_q(y)$.
\end{enumerate}
\end{lemma}

Thus, by constructing a basis of the Riemann-Roch space $\cL(\lambda \fP_{\infty})$ and selecting a certain multipoint set as in Section~3, we can define \rAG codes from $E$ that care fast encodable and burst decodable.

\subsubsection{ \textbf{$\F_q$ has constant characteristic $p$}}\label{sec:hcca}
Assume $\kappa=p^r$. Let $\{1=w_1,w_2,\dots,w_{r}\}$ be an $\F_p$-basis of $\F_{\kappa}$. For $1\leq i\leq r$, let $W_{i}=\oplus_{j=1}^i\F_pw_j$ be an $\F_p$-subspace of dimension $i$. By Lemma~\ref{lem:preH}(2), there exists a subgroup
\[G_i\leq\Gal(E/\F_{q}(x))\ \text{such\ that}\ G_i\cong W_i.\] 
Then the fixed subfield $E_i=E^{G_i}=\F_q(y_i,x)$ and
$y_i=\prod_{\alpha\in W_i}(y+\alpha)$~\cite{encAG24}. Thus $E$ has a chain of subfields: $E_i\supset E_{i+1}$ with extension degree $[E_i: E_{i+1}]=p$, for $i=1,\dots,r-1$.

Next we construct the basis of $\cL(\lambda{\fP}_{\infty})$ as in Lemma~\ref{lem:basis}. For any $0\leq u\leq p^r-1$, let $\bu=(u_0,u_1,\dots,u_{r-1})$ be the coefficients of $p$-adic expansions of $u=\sum_{i=0}^{r-1}u_ip^i$. Define
\begin{equation}\label{eq:hexp}
    \by^{\bu}=y_0^{u_0}y_1^{u_1}\cdots y_{r-1}^{u_{r-1}}.
\end{equation}
Since $\nu_{\fP_{\infty}}(\by^{\bu})=\sum_{i=0}^{r-1}u_i\nu_{\fP_{\infty}}(y_i)=-\sum_{i=0}^{r-1}u_ip^i(\kappa+1)=\nu_{\fP_{\infty}}(y^{u})$, by substituting $x^jy^u$ with $x^j\by^{\bu}$ in $\cB_1$ in Lemma~\ref{lem:preH}(1), the set $\tcB_1=\{x^{j}\by^{\bu}\mid j\kappa+u(\kappa+1)\leq \lambda, 0\leq u\leq \kappa-1\}$ is also a basis of $\cL(\lambda{\fP}_{\infty})$.

By Lemma~\ref{lem:preH}(2), any rational place $P\neq P_{\infty}$ in $\F_q(x)$ splits into $\kappa$ rational places in $E$. For different purposes of decoding, we first choose $\{P_1,P_2,\dots, P_n\}\subset \PP_{\F_q(x)}$ as follows and take $\calP$ be the the set of rational places in $E$ lying over these $n$ rational places. 
\begin{itemize}
\item[(i)] For list decoding, we take $n=q=\kappa^2$ and $\{P_1,P_2,\dots, P_n\}=\F_q$. Let $\calP$ be the set of rational places in $E$ lying over $\{P_1,P_2,\dots, P_n\}$. Then $N=|\calP|=\kappa^3$.
\item[(ii)] For probabilistic unique decoding, we take $n=\ell^{\frac{2r}{s}-1}(\ell-1)<q$ for $\ell=p^s$ is a small constant power of $p$. Let $\{P_1,P_2,\dots, P_n\}=\{u(\gamma+w)\mid u\in\F_{\ell}^*,\ L(w)=0\}$, where $L(x)$ is a linearized polynomial of degree $\ell^r$ and $\gamma\in\F_q$ satisfying $L(\gamma)\neq 0$ as in Lemma~\ref{lem:MPset}. Thus $N=|\calP|=\kappa^3(1-\frac{1}{\ell})$.
\end{itemize}
From above, all assumptions (P1)-(P4) hold for the Hermitian function field. Assume $k$ is a positive integer such that $\kappa\leq k\leq n$. Let $\lambda=\kappa(k-1)$ and $\fC(\lambda \fP_{\infty},\calP)$ be the Hermitian code defined on $\calP$. Since $\lambda\geq 2\fg(E)-1=\kappa(\kappa-1)-1$, by the Riemann-Roch theorem \cite{sti09}, $\fC(\lambda \fP_{\infty},\calP)$ has code rate 
$$\cR=\frac{k-1}{n}+\frac{1}{n\kappa}-\frac{\kappa-1}{2n}.$$ Moreover, from the above construction, the constituent $\RS$ code in the $\AG$-based $\IRS_{m}(\fC(\lambda \fP_{\infty},\calP))$ satisfies assumptions in Section~\ref{sec:RS}. By Theorem~\ref{thm:LDAG} and \ref{thm:uniqDecAG}, we have the following results

\begin{corollary}\label{coro:ha}
    Assume $\F_q$ has constant characteristic. Let $\fC(\lambda \fP_{\infty},\calP)$ be defined as above and $\cR$ be the code rate.  
\begin{itemize}
\item[(1)] For any real number $\ep>0$, $\fC(\lambda \fP_{\infty},\calP)$ is $(1-\cR-\ep,O(1/\ep))$-burst list decodable. Moreover, it can be list decoded in time $O_{\ep}(N\log N)$.
\item[(2)] For any real number $\ep>0$ and positive integer $e\geq 2$, $\fC(\lambda \fP_{\infty},\calP)$ can correct bursts of length less than $N\left(1-\cR-\ep\right)$ with miscorrection probability at most $\ep\kappa/q^{e-1}=\ep/\kappa^{2e-3}$. Moreover, it can be probabilistic uniquely decoded in time $O(\frac{e}{\ep}N\log N)$.
\end{itemize}
\end{corollary}
\begin{proof}
By Theorem~\ref{thm:LDAG}, for any $\ep>0$, we have that $\fC(\lambda \fP_{\infty},\calP)$ is $(1-\cR-\fg/N-\ep/2,O(1/\ep))$-burst list decodable. Since $N=\kappa^3$, $\fg/N=\frac{\kappa(\kappa-1)}{2\kappa^3}<1/2\kappa<1/2n_s$, where $n_s=2/\ep$. Thus $1-\cR-\ep\leq 1-\cR-\fg/N-\ep/2$ and the results in (1) follow immediately.

By Theorem~\ref{thm:uniqDecAG}, for any $\ep>0$ and constant integer $e\geq 2$, we have that $\fC(\lambda \fP_{\infty},\calP)$ can correct burst errors of length less than $N(1-\cR-\fg/N-\ep/2)$ with miscorrection probability at most $\ep/\kappa^{2e-3}$. Put $\ell=2/\ep$ and $\kappa=\ell^{r/s}>\ell^2$. By settings in (ii), $\fg/N<1/2\kappa(1-1/\ell)<\ep/2$. Then the results in (2) follow.
\end{proof}

\subsubsection{\textbf{$q-1$ is $O(1)$-smooth.}} \label{sec:hccm}
Assume $\kappa+1=\prod_{i=1}^r p_j\mid q-1$. By Lemma~\ref{lem:preH}(3), let 
\[G_i\leq \Gal(E/\F_q(y))\ \text{of\ order}\ |G_i|=\prod_{j=1}^ip_j\ text{for}\ 1\leq i\leq r.\]
Then the fixed subfield $E_i=E^{G_i}=\F_q(x_i,y)$, where $x_i=\prod_{\alpha\in G_i}\alpha (-x)=x^{|G_i|}$. Thus $E$ has a chain of subfields: $E_i\subset E_{i+1},$ and $[E_i: E_{i+1}]=p_i$. 

For any $0\leq v\leq \kappa$, let $\bv=(v_0,v_1,\dots,v_{r-1})$ be the coefficients of expansions of $v=\sum_{i=0}^{r-1}v_{i-1}p_0p_1\cdots p_i$ (where $p_0:=1$) as in Equation~\eqref{eq:expu}. Then 
$$\bx^{\bv}=x_0^{v_0}x_1^{v_1}\cdots x_{r-1}^{v_{r-1}}=x^v.$$ Hence, the basis $\{\bx^{\bv}y^{j}\mid j(\kappa+1)+v\kappa\leq \lambda, 0\leq v\leq \kappa\}$ in Lemma~\ref{lem:basis} is just $\cB_2$.

Let $\Ker(\Tr_{q/\kappa})=\{\alpha\in \F_q \mid \alpha^{\kappa}+\alpha=0\}$. By Lemma~\ref{lem:preH}(3), for any $\beta\notin \Ker(\Tr_{q/\kappa})$, the equation $x^{\kappa+1}=\beta^{\kappa}+\beta$ has $\kappa+1$ solutions in $\F_q$, i.e., $y-\beta$ splits completely in $E/\F_q(y)$. We take a multiplicative group or a coset of a multiplicative group $\sP$ in $\F_q\setminus\Ker(\Tr_{q/\kappa})$ of order $n$. Let $\calP$ be the set of all rational places lying above $\sP$. Then $N=|\calP|=n(\kappa+1)$ and $N$ is $O(1)$-smooth.

Assume $\kappa\leq k\leq n$ is a positive integer. Let $\lambda=(k-1)(\kappa+1)$ and $\fC( \lambda{\fP}_{\infty},\calP)$ be the Hermitian code defined on $\calP$. Then the code rate is 
$$\cR=\frac{k-1}{n}+\frac{1}{n(\kappa+1)}-\frac{\kappa(\kappa-1)}{2n(\kappa+1)}.$$  Since $n\mid q-1$ is $O(1)$-smooth, the RS code $\RS[n,k]$ defined on the multipoint set $\sP$ satisfies conditions in Theorem~\ref{thm:listDecRS} and \ref{thm:uniqDecRS}. 
\begin{corollary}
   Assume $q-1$ is $O(1)$-smooth. Let $\fC( \lambda{\fP}_{\infty},\calP)$ be defined as above and $\fg=\frac{\kappa(\kappa-1)}{2}$ be the genus of $E$. For any $\ep>0$,
\begin{itemize}
\item[(1)] $\fC(\lambda \fP_{\infty},\calP)$ is $(1-\cR-\fg/N-\ep,O(1/\ep))$-burst list decodable; 
\item[(2)] $\fC(\lambda \fP_{\infty},\calP)$ can correct bursts of length less than $N\left(1-\cR-\fg/N-\ep(2+e)\right)$ for $e\geq 2$ with miscorrection probability at most $\ep(\kappa+1)/q^{e-1}$. 
\end{itemize}
Moreover, there exist $O_{\ep}(N\log N)$ list and probabilistic unique decoding algorithms for $\fC( \lambda{\fP}_{\infty},\calP)$.
\end{corollary}

\end{example}

\subsection{Examples of {$\AG$} codes from Hermitian tower}
Let $\kappa$ be a prime power and let $q=\kappa^2$. In this example, we only consider the case that $\F_q$ has constant characteristic. The case that $q-1$ is $O(1)$-smooth can be similarly discussed.

Let $F_1=\F_q(x_1)$. The Hermitian tower \cite{Shen93} is defined by the following recursive equations
\begin{equation}\label{eq:x2}
x_{i+1}^{\kappa}+x_{i+1}=x_i^{{\kappa}+1},\quad i=1,2,\dots,t-1.
\end{equation}
Put  $F_i=\F_q(x_1,x_2,\dots,x_{i})$ for $i\ge 2$. We fix an integer $t$ satisfying $2\le t\le \kappa/2$. 

Since each field extension $F_i/F_{i-1}$ has the same properties as Hermitian function field $F_2/F_1$ in Lemma~\ref{lem:preH}, we have the following facts:
\begin{itemize}
\item[(i)] The pole $P_{\infty}\in\PP_{\F_q(x_1)}$ of $x_1$ is totally ramified in the extension $F_t/F_1$. Let $\fP_{\infty}^{(t)}$ be the unique place of $F_t$ lying over $P_{\infty}$. For an integer $\Gl>0$, the Riemann-Roch space $\cL(\Gl \fP_\infty^{(t)})$ of $F_t$ has a basis as follows \cite[Proposition 5]{Shen93}
\begin{equation}\label{eq:x3}
\small
\cB_1^{(t)}=\left\{x_1^{j_1}\cdots x_t^{j_t}\mid (j_1,\dots,j_t)\in\mathbb{N}^t,\ \sum_{i=1}^tj_i\kappa^{t-i}({\kappa}+1)^{i-1}\le \Gl, 0\le j_2,\cdots,j_t\le \kappa-1\right\}
\end{equation}
\item[(ii)] For every $\Ga\in\F_q$, the rational place $P_\Ga\in\PP_{\F_q(x_1)}$ splits completely in $F_t$, i.e., there are $\kappa^{t-1}$ rational places in $F_t$ lying over $P_\Ga$.
\end{itemize}

Replace each $x_i^{j_i}$ with $\bx_i^{\bj_i}$ as in Equation~\eqref{eq:hexp} for $i=2,\dots,t$, then $\cB_1^{(t)}=\{x_1^{j_1}\bx_2^{\bj_2}\cdots \bx_t^{\bj_t}$ is a basis of $\cL(\Gl \fP_\infty^{(t)})$ ensuring that the G-FFT and G-IFFT work well for functions in $\cL(\lambda_t{\fP}^{(t)}_{\infty})$. Next, let $\{P_1,P_2,\dots, P_n\}\subsetneq {\F_q}$ be the same multipoint set chosen as in \S~\ref{sec:hcca} (i) and (ii), i.e., $n=\kappa^2$ for list decoding, and $n=\kappa^2(1-\ep)$ for unique decoding, respectively, and $\calP^{(t)}$ be the set of rational places in $F_t$ lying over them. Then $|\calP^{(t)}|:=N_t=\kappa^{t-1} n$. Let $\kappa\leq k\leq n$, $\lambda_t=\kappa^{t-1}(k-1)$ and $\fC(\lambda_t \fP^{(t)}_{\infty},\calP)$ be the Hermitian code defined on $\calP^{(t)}$. Then the code rate
\[\cR_t=\frac{(k-1)\kappa^{t-1}+1-\fg_t}{n\kappa^{t-1}}\geq \frac{k-1}{n}-\frac{t\kappa}{n}.\]
The above inequality follows from the genus $\fg_t$ of $\F_t$ satisfying $\fg\leq t\kappa^t$ \cite[Section 6]{encAG24}. Hence, for $t<\kappa \ep$, we have $\fg_t/N_t<\ep$. By similar analysis as in Corollary~\ref{coro:ha}, we have the following results for Hermitian tower codes. 
\begin{corollary}\label{coro:ha}
    Assume $\F_q$ has constant characteristic. Let $\fC(\lambda_t \fP_{\infty}^{(t)},\calP^{(t)})$ be defined as above and $\cR_t$ be the code rate. For any real number $\ep>0$, let $t<\kappa\ep$. 
\begin{itemize}
\item[(1)] $\fC(\lambda_t \fP_{\infty}^{(t)},\calP^{(t)})$ is $(1-\cR_t-\ep,O(1/\ep))$-burst list decodable. Moreover, it can be list decoded in time $O_{\ep}(N\log N)$.
\item[(2)] For any positive integer $e>(t+1)/2$, $\fC(\lambda_t \fP^{(t)}_{\infty},\calP^{(t)})$ can correct bursts of length less than $N_t(1-\cR_t-\ep)$ with miscorrection probability at most $\ep/\kappa^{2e-t-1}$. Moreover, it can be probabilistic uniquely decoded in time $O(\frac{e}{\ep}N\log N)$.
\end{itemize}
\end{corollary}

\subsection{Examples of {$\AG$} codes from the Garcia-Stichtenoth tower}
As we have mentioned (P1)-(P3) in the G-FFT assumptions are used to construct a basis for the one-point Riemann-Roch space, from which we can reduce the decoding problem of AG codes to the decoding problem of interleaved RS codes. There are some function fields $E/\F_q$ that do not satisfy all three assumptions (P1)-(P3) of the G-FFT, but have  good polynomial base. By using the same technique as in previous sections, we will show that the algenraic geometry codes based on $E$ can still be burst decoded in quasi-linear time.

In the following, we always $\F_q$ has constant characteristic $\Char \F_q=O(1)$.
\begin{example}{(Garcia-Stichtenoth tower)}\label{ex:gs}
Let $\F_q=\F_{\kappa^2}$ be a finite field and $\F_q(x)$ be the rational function field. The Garcia-Stichtenoth tower over $\F_q$ is the sequence of function fields defined by $F_1=\F_q(x),\ F_{i+1}=F_i(z_{i+1})$ where $z_{n+1}$ satisfies
\begin{equation}\label{eq:defGS}
    z^{\kappa}_{n+1}+z_{n+1}=x_n^{\kappa+1},\ x_1=x\ \text{and}\ x_n=\frac{z_n}{x_{n-1}}\ \text{for}\ n\geq 2.
\end{equation}
Particularly, $F_2$ is just the Hermitian function field. \textbf{Note that the field extension $F_n/\F_1$ does not satisfy  (P3) in the G-FFT assumptions in section 3}. In \cite{gs95tower}, the authors pointed out the extension of places in $F_n/F_1$. We will only consider the third function field $F_3$ in the following.
\begin{itemize}
    \item[(i)] All places in $\{P_{\alpha}\in\PP_{F_1}\mid \alpha\in\F_q^*\}$ are splitting completely in $F_3/F_1$: first, the equation $z_2^{\kappa}+z_2=\alpha^{\kappa+1}$ has $\kappa$ solutions in $\F_q$; next, for each solution $z_2=\beta$, the equation $z_3^{\kappa}+z_3=(\frac{\beta}{\alpha})^{\kappa+1}$ also has $\kappa$ solutions. We denote $P_{\alpha,\beta,\gamma}$ by the zero place of $x_1=\alpha,z_2=\beta,z_3=\gamma$. There are $(\kappa^2-1)\kappa^2$ such places $\{P_{\alpha,\beta,\gamma}\in\PP_{F_3}\mid \alpha\in\F_q^*\}$ in total.
    \item[(ii)] Let $P^{(1)}_{\infty}\in\PP_{F_1}$ be the unique pole place of $x_1$ in $F_1$. Then $P^{(1)}_{\infty}$ is totally ramified in $F_3/F_1$, denoted by $P_{\infty}^{(3)}$ the unique place in $F_3$ lying over $P^{(1)}_{\infty}$.
\end{itemize}
Moreover, the genus of $F_3$ is $\fg(F_3)=\kappa^3-2\kappa+1$. Consider the Riemann-Roch space $\cL(\lambda P^{(3)}_{\infty})$ with $\lambda=(k-1)q$ and the positive integer $k$ satisfying $2\kappa<k<\kappa^2$. Then 
\[\dim_{\F_q}\left(\cL(\lambda P^{(3)}_{\infty})\right)=(k-1-\kappa)\kappa^2+2\kappa.\]
For the third function field $F_3$, by~\cite{vos97GS3}, we can get an explicit basis of the Riemann-Roch space $\cL(\lambda P^{(3)}_{\infty})$. By computation, we note that there is a subspace of $\cL(\lambda P^{(3)}_{\infty})$ which has dimension close to $\dim_{\F_q}\cL(\lambda P^{(3)}_{\infty})$ but a simpler basis. Under this basis, we can use the same technique to burst decode the \rAG codes based on the third function $F_3$.

\begin{lemma}\label{lem:subsp}
   Define the following set
    \[\begin{split}
        \mathcal{I}_1(\lambda)=\{(i_1,i_2,i_3)\mid\  &i_1\kappa^2+i_2\kappa+i_3(\kappa+1)\leq \lambda,\ i_2\kappa+i_3(\kappa+1)\leq i_1\kappa,\\
   &0\leq i_1,\ 0\leq i_2,i_3\leq \kappa-1\}.
    \end{split}  \]
   The set $\cB(\lambda)=\{x_1^{i_1}x_2^{i_2}z_3^{i_3}\mid (i_1,i_2,i_3)\in\mathcal{I}_1(\lambda)\}$ is an $\F_q$-linearly subset of $\cL(\lambda P^{(3)}_{\infty})$. Moreover, 
   \[\dim_{\F_q}(\cL(\lambda P^{(3)}_{\infty}))-|\cB(\lambda)|\leq (\kappa^2+\kappa-2)/2.\]
\end{lemma}
\begin{proof}
    Define 
    \[\begin{split}\mathcal{I}_2(\lambda)&=\{(i_1,-i_2,i_3)\mid\  i_1\kappa^2+i_2\kappa+i_3(\kappa+1)\leq \lambda,\ i_3(\kappa+1)\leq i_2\kappa+i_1\kappa, \\
   &1\leq i_2\leq \kappa,\ i_1\geq \kappa i_2,\ 0\leq i_3\leq \kappa-1,\ (i_1-1,\kappa-i_2,i_3)\notin \mathcal{I}_1(\lambda)\}.
\end{split}\]
    The Lemma~4.2 in \cite{vos97GS3} shows that $\cB'(\lambda)=\{x_1^{i_1}x_2^{i_2}z_3^{i_3}\mid (i_1,i_2,i_3)\in\mathcal{I}_1(\lambda)\cup\mathcal{I}_2(\lambda)\}$ is a basis of $\cL(\lambda P^{(3)}_{\infty})$. It suffice to show $|\mathcal{I}_2(\lambda)|\leq (\kappa^2+\kappa-2)/2=O(q)$. Since
    \[\nu_{P_{\infty}^{(3)}}(x_1^{i_1}x_2^{-i_2}z_3^{i_3})=\nu_{P_{\infty}^{(3)}}(x_1^{i_1-1}x_2^{\kappa-i_2}z_3^{i_3}),\]
    by the definition of $\mathcal{I}_1(\lambda)$, the condition $(i_1-1,\kappa-i_2,i_3)\notin \mathcal{I}_1(\lambda)$ implies that
    \begin{equation}\label{eq:I2}
    (\kappa-i_2)\kappa+i_3(\kappa+1)>(i_1-1)\kappa\Rightarrow \kappa^2+\kappa+i_3(\kappa+1)>i_1\kappa+i_2\kappa.
    \end{equation}
    As $i_1\geq i_2\kappa$, $i_2\geq 1$ and $0\leq i_3\leq \kappa-1$ in $\mathcal{I}_2(\lambda)$, there must be 
    \[1\leq i_2<1+i_3/\kappa \Rightarrow i_2=1,\ 1\leq i_3\leq\kappa-1,\ \text{and}\ \kappa\leq i_1.\]
    In this case, the inequality in Equation~\eqref{eq:I2} implies
    \[\kappa \leq i_1\leq i_3+\kappa,\text{and}\ 1\leq i_3\leq \kappa-1.\]
    There are at most $\sum_{i_3=1}^{\kappa-1}(i_3+1)=(\kappa^2+\kappa-2)/2$ such tuples $(i_1,i_2=-1,i_3)\in \mathcal{I}_2(\lambda)$.
\end{proof}

Let $K=|\cB(\lambda)|$ and denote the elements in $\cB(\lambda)$ simply by $\cB(\lambda)=\{\mathbf{b}_i:\ i=1,\dots,K\}$. We define the \rAG code from $F_3$ as follows:
\begin{itemize}
    \item\textbf{Message space}: $\mathcal{V}=\left\{\sum_{i=1}^K\alpha_i\mathbf{b}_i\mid \alpha_i\in\F_q\right\}\cong\F_q^K$ and $K\approx (k-1-\kappa)\kappa^2$ by Lemma~\ref{lem:subsp};
    \item\textbf{The set of evaluation places}: Assume $\F_q=\F_{\ell^{u}}$ and $W\subset\F_q$ is a subspace of dimension $u-1$, where $\ell\mid q$ is a constant prime power. Let $L(x)=\prod_{\alpha\in W}(x-\alpha)$ and $\gamma\in\F_q$ such that $L(\gamma)\neq 0$. We take a subset $\sP=\{\alpha_1,\dots,\alpha_n\}\subset \F_q^*$ of order $n=(\ell-1)\ell^{u-1}=q(1-1/\ell)$ which are the roots of $L(x)^{\ell-1}=L(\gamma)^{\ell-1}$ as in Lemma~\ref{lem:MPset}. Let $\calP=\cup_{i=1}^n\calP_{\alpha_i}$, where $\calP_{\alpha_i}=\{P_{\alpha_i,\beta,\gamma}:\ x(P_{\alpha_i,\beta,\gamma})=\alpha_i\}$ is the set of all places lying over $P_{\alpha_i}$ defined as in (i). Then $|\calP|=n\kappa^2=\kappa^4(1-1/\ell)$. 
\end{itemize}
The \rAG code from $\GS$-tower is defined as
\[\fC(\mathcal{V},\calP)=\left\{\ev_{\calP}(f)=\left(\ev_{\calP_{\alpha_1}}(f),\dots,\ev_{\calP_{\alpha_{n}}}(f)\right)\mid f\in\mathcal{V}\right\}.\]
Thus $\fC(\mathcal{V},\calP)$ is an $\F_q$-linear code with length $N=n\kappa^2$ and dimension $K=(k-1-\kappa)\kappa^2$, where $n<\kappa^2$ and $2\kappa+1\leq k\leq n$. 

Every function $f\in\mathcal{V}$ under the basis $\cB(\lambda)$ can be decomposed as follows:
\[f(x_1,x_2,z_3)=\sum_{i_2,i_3=1}^{\kappa-1}f_{i_2,i_3}(x)x_2^{i_2}z_3^{i_3},\ f_{i_2,i_3}(x)\in\F_q[x]_{<k}
\]
Since $\F_q$ has constant characteristic, by the G-FFT technique in Section~\ref{sec:RS}, put $m_s=\ell^{u-1}$, $n_s=\ell-1$ and $x_s=L(x)$, then $f(x_1,x_2,z_3)$ can be further decomposed as 
\[ f(x_1,x_2,z_3)=\sum_{j_s=0}^{m_s-1}\sum_{i_2,i_3=1}^{\kappa-1}f_{j_s,i_2,i_3}(x_s)x_1^{j_s}x_2^{i_2}z_3^{i_3},\ f_{j_s,i_2,i_3}(x_s)\in\F_q[x_s]_{<\lceil k/m_s\rceil}.
\]
Let $\sP_s=\{t_1,\dots,t_{\ell-1}\}=\{x_s(\alpha)\mid \alpha\in\sP\}$ be the set of rational places in $\F_q(x_s)$ lying below $\calP$. For each $t_j$, there are $m_s\kappa^2$ places in $\calP$ lying over $x_s-t_j$, i.e., 
$$\calP[j]:=\{(\alpha, \beta,\gamma)\ \text{satisfying\ Equation}\ \eqref{eq:defGS}\mid \alpha=t_jw,\ w\in W\}.$$
Then 
\[\ev_{\calP[j]}(f)=\ev_{\calP[j]}\left(\sum_{j_s=0}^{m_s-1}\sum_{i_2,i_3=1}^{\kappa-1}f_{j_s,i_2,i_3}(t_j)x_1^{j_s}x_2^{i_2}z_3^{i_3}\right).\]
For each polynomial $f_j:=\sum_{j_s=0}^{m_s-1}\sum_{i_2,i_3=1}^{\kappa-1}f_{j_s,i_2,i_3}(t_j)x_1^{j_s}x_2^{i_2}z_3^{i_3}\in\F_q[x_1,x_2,z_3]$ with three variables, we make use of the fast multivariate MPE and interpolation \cite{BRS20,von13} to do the DFT and IDFT, respectively. Then the multipoint evaluation of $f_j$ at $\calP[j]$ can be computed in time $\tO(m_s\kappa^2)$; conversely, given the MPE of $f_j$ at $\calP[j]$, the coefficients $\{f_{j_s,i_2,i_3}(t_j) \mid 0\leq j_s\leq m_s-1,\ 0\leq i_2,i_3,\leq \kappa-1\}$ can be computed in time $\tO(m_s\kappa^2)$. Thus by folding $\ev_{\calP}(f)$ with length $m_s\kappa^2$ and running multivariate IDFT on each column, we get an interleaved $\RS[n_s,k/m_s]$
\[
(\ev_{\calP_{\alpha_1}}(f),\ldots,\ev_{\calP_{\alpha_{n}}}(f))
\xrightleftharpoons[]{}
\left(\begin{matrix}
f_{0,0,0}(t_1)                   &   \cdots   & f_{0,0,0}(t_{n_s})\\ 
f_{1,0,0}(t_1)                   &   \cdots   & f_{1,0,0}(t_{n_s})\\ 
\vdots                           &   \cdots   & \vdots            \\
f_{m_s-1,\kappa-1,\kappa-1}(t_1)  &   \cdots   & f_{m_s-1,\kappa-1,\kappa-1}(t_{n_s})\\  \end{matrix}\right).
\]
Thus decoding $\fC(\mathcal{V},\calP)$ with burst errors can be reduced decoding $\IRS[n_s,k/m_s]$ with burst errors. By a similar proof of Corollary~\ref{coro:ha}, we have the following results
\begin{itemize}
\item[(1)] For any $\ep>0$, $\fC(\mathcal{V},\calP)$ is $(1-\cR-\ep,O(1/\ep))$-burst list decodable; 
\item[(2)] For any  $\ep>0$ and constant integer $e\geq 2$, $\fC(\mathcal{V},\calP)$ can correct bursts of length less than $N(1-\cR-\ep)$ with miscorrection probability at most $\ep/q^{e-2}$,
\end{itemize}
where $\cR=K/N$ is the code rate. Moreover, there exist quasi-linear time $\tO_{\ep}(N)$ list and probabilistic unique decoding algorithms for $\fC( \lambda{\fP}_{\infty},\calP)$.
\end{example}

\bibliographystyle{alpha}
\bibliography{sjtu_li24}

\appendix
\section{The proof of Lemma~\ref{lem:UniDec}}\label{app:uniq}
In this section, we will give a generalized version of Wu's algorithm that can be applied to $\RS$ codes defined on a non-cyclic multipoint set. Let us first recap some notations in syndrome decoding. 

Assume $\calP=\{\xi,\xi\alpha,\dots,\xi\alpha^{n-1}\}\in\F_q^*/\langle \alpha\rangle$ is a coset of cyclic group, where $\alpha\in\F_q^*$ has order $n$. The $\RS[n,k]$ that we concerned is defined as 
\[\RS[n,k]=\{\left(f(\xi),f(\xi\alpha^{}),\dots,f(\xi\alpha^{n-1})\right)\mid f(x)\in\F_q[x]_{<k}\}.\]
Let $r=n-k$. For a burst interval $[i,i+r-2]$, we always view it as a subset of $[1,n]$ in the following sense
$$[i,i+r-2](\bmod\ {n+1})=\begin{cases}
    [i,i+r-2],\ &\text{if}\ i+r-2\leq n;\\
    [i,n]\cup [1,(i+r-2\bmod n+1)+1],\ &\text{if}\ i+r-2> n.
\end{cases}$$
The corresponding reciprocal error locator of $[i,i+r-2]$ is defined as
\begin{equation}\label{eq:errorlocator}
    \Lambda^{(i)}(x)=\prod_{j=i}^{i+r-2}(1-\xi\alpha^{j-1}x)=\Lambda^{(i)}_0+\Lambda^{(i)}_1x+\dots,\Lambda^{(i)}_{r-1}x^{r-1}.
\end{equation}
Given a received word $\br=\bc+\be$, where $\be$ is a burst with length $\fl\leq r-1$. The syndrome polynomial of $\br$ is defined as 
\begin{equation}\label{eq:syndpoly}
    \begin{split}S(x)&=\sum_{j=1}^{n} \frac{r_j\Gd_j}{1-\xi\Ga^{j-1}x}\bmod x^{n-k},\ \text{where}\ \Gd_j=\prod_{k\neq j}(\xi\alpha^{j-1}-\xi\alpha^{k-1})^{-1}\\
&=S_0+S_1x+\dots+S_{r-1}x^{r-1}.\end{split}
\end{equation}

Assume the burst interval of $\be$ is $[i,i+\fl-1]$. Then $[i,i+\fl-1]\subset [i,i+r-2]$. By syndrome decoding \cite{10synd}, we have
\begin{equation}\label{eq:synd}
S_{r-1}\Lambda_{0}^{(i)}+S_{r-2}\Lambda_{1}^{(i)}+\dots+S_1\Lambda_{r-2}^{(i)}+S_0\Lambda_{r-1}^{(i)}=0.
\end{equation}
Particularly, let $\Lambda^{(1)}(x)=\prod_{j=1}^{r-1}(1-\xi\alpha^{j-1}x)$ be the reciprocal error locator of burst $[1,r-1]$. Define a polynomial as follows
\begin{equation}\label{eq:gamax}
\Gamma(x)=S_{r-1}\Lambda_{0}^{(1)}+S_{r-2}\Lambda_{1}^{(1)}x+\dots+S_0\Lambda_{r-1}^{(1)}x^{r-1}.
\end{equation}

\begin{lemma}\label{lem:app1}
Assume $\br$ is a corrupted codeword in $\RS[n,k]$ by an $\fl$-burst with interval $[i,i+\fl-1]\subset[1,n]$. Let $S(x)$ be the syndrome polynomial of $\br$ and $\Gamma(x)$ be defined as in Equation~\eqref{eq:gamax}. If $\fl\leq n-k-1$, then $\Gamma(x)$ has $n-k-1-\fl$ consecutive roots $\alpha^{i-1},\alpha^{i-2},\dots,\alpha^{i-(n-k)+\fl}$.
\end{lemma}
\begin{proof}
Let $r=n-k$. By definition, $\Lambda^{(i)}(x)=\Lambda^{(1)}(\alpha^{i-1}x)=\Lambda_{0}^{(1)}+\Lambda_{1}^{(1)}\alpha^{i-1}x+\dots+\Lambda_{r-1}^{(1)}(\alpha^{i-1}x)^{r-1}$. Thus the coefficients satisfy
\[\Lambda_{j}^{(i)}=\Lambda_{j}^{(1)}\alpha^{j(i-1)}\ \text{for}\ j=0,\dots,r-1.\]
By substituting the above relation into Equation~\eqref{eq:synd}, we have
\[0=S_{r-1}\Lambda_{0}^{(1)}+S_{r-2}\Lambda_{1}^{(1)}\alpha^{i-1}+\dots+S_0\Lambda_{r-1}^{(1)}\alpha^{(r-1)(i-1)}=\Gamma(\alpha^{i-1}).\]
For $u=i,i-1,\dots,i-(r-\fl)+1$, let $\Lambda^{(u)}(x)$ be the reciprocal error locator of $[u,u+r-2]$ defined as in Equation~\eqref{eq:errorlocator}. Since the burst interval $[i,i+\fl-1]$ is also contained in the interval $[u,u+r-2]$,
the coefficients of $\Lambda^{(u)}(x)$ and $S(x)$ also satisfy Equation~\eqref{eq:synd}, i.e.,
\[\sum_{j=0}^{r-1}S_{r-1-j}\Lambda_{j}^{(u)}(\alpha^{j(u-1)})=\Gamma(\alpha^{u-1})=0.\]
Thus, $\alpha^{i-1},\alpha^{i-2},\dots,\alpha^{i-(n-k)+\fl}$ are $\fl$ consecutive roots of $\Gamma(x)$.
\end{proof}

 \begin{algorithm}[!h]\label{alg: UniqDec}
  \caption{Burst probabilistic unique decoding of cyclic RS[n,k].}
  \label{alg:UniqDec}
  \begin{algorithmic}[1]
    \Require
    A corrupted RS codeword $\br\in\F_q^n$ by a burst error of length $\leq n-k-1$.
    \Ensure
      $(c_1,c_2,\dots,c_n) \in\RS[n,k]$.
    \State Precompute $\Lambda^{(1)}(x)=\prod_{j=1}^{n-k-1}(1-\xi\alpha^{j-1}x)=\Lambda^{(1)}_0+\Lambda^{(1)}_1x+\dots,\Lambda^{(1)}_{n-k-1}x^{n-k-1}$.
    \State Compute the syndrome polynomial $S(x)=S_0+S_1x+\dots+S_{n-k-1}x^{n-k-1}$ of $\br$.
    \State Compute $\Gamma(x)=\sum_{j=0}^{n-k-1}S_{n-k-1-j}\Lambda^{(1)}_{j}x^j$.
    \State Compute the roots of $\Gamma(x)$ among $\{1,\Ga, \dots,\Ga^{n-1}\}$. Find the longest consecutive root sequence $\alpha^{i-1},\dots,\alpha^{i-(n-k)+\fl}$. Record the start position $i$ and the length $\fl$.
     \State Erasure decoding $\br$ with erased positions $[i,i+\fl-1]$:
              \begin{enumerate}
                 \item[(i)] Compute $\lambda(x)=\prod_{j=i}^{i+\fl-1}(x-\xi\alpha^{j-1})$ and its derivative $\lambda'(x)$.
                 \item[(ii)] Using inverse FFT to compute $F(x)\in\F_q[x]_{<n}$ such that $F(\xi\alpha^{j-1})=\lambda(\xi\alpha^{j-1})r_j$ if $j\notin[i,i+\fl-1]$, and $F(\xi\alpha^{j-1})=0$ if $j\in[i,i+\fl-1]$. Moreover, compute the derivative $F'(x)$ of $F(x)$.
                 \item[(iii)] Compute $c_j=F'(\xi\alpha^{j-1})\lambda(\xi\alpha^{j-1})^{-1}$ for $j\in[i,i+\fl-1]$ and $c_j=r_j$ for $j\notin[i,i+\fl-1]$.
              \end{enumerate}
  \State \Return  $(c_1,c_2,\dots,c_n)$.
  \end{algorithmic}
\end{algorithm}

\begin{lemma}{\cite{wu12}}
Let $r=n-k$. Then Algorithm 3 can correct $\RS[n,k]$ with burst errors of length $\fl\leq r-1$ with miscorrection probability at most $1/q^{r-1-{\fl}}$. Moreover, it can be decoded in time $O(n\log n)$ if $n$ is $O(1)$-smooth. 
\end{lemma}
\begin{proof}
According to Lemma~\ref{lem:app1}, we can find the longest consecutive roots of $\Gamma(x)$ via Steps 1-4, which corresponds to the shortest burst $\be$ in $\{\RS[n,k]-\br\}$. Assume $[i,i+\fl-1]$ is the burst interval of $\be$. Let $f(x)\in\F_q[x]_{<k}$ be the correct message polynomial such that $f(\xi\alpha^{j-1})=r_j$ for any $j\in[1,n]\setminus[i,i+\fl-1]$. Let $\lambda(x)$ be the erasure error locator defined as in Step 5(i) and $F(x)$ be the interpolate polynomial such that
\[F(\xi\alpha^{j-1})=\begin{cases}
    f(\xi\alpha^{j-1})\lambda(\xi\alpha^{j-1}),\ &\text{if}\ j\notin[i,i+\fl-1];\\
    0,\ &\text{if}\ j\in[i,i+\fl-1].
\end{cases}\]
Then 
\[F(x)=\lambda(x)f(x)\Rightarrow F'(x)=\lambda'(x)f(x)+\lambda(x)f'(x).\]
Thus $c_j=f(\xi\alpha^{j-1})=F'(\xi\alpha^{j-1})/\lambda'(\xi\alpha^{j-1})$ for $j\in [i,i+\fl-1]$. Hence the output $\bc$ is a codeword such that $\br-\bc$ is an $\fl$-burst.

The miscorrection probability is the probability that $\Gamma(x)$ has spurious consecutive roots $\alpha^{i'-1},\dots,\alpha^{i'-r+\fl'}$ with $\fl'\leq \fl$. By~\cite[Theorem 2]{wu12}, the probability is at most $1/q^{r-1-\fl}$.

For decoding complexity, Step 2 costs $O(n\log n)$ operations in $\F_q$ by \cite{von13} and Step 3 costs at most n operations. Step 4 costs $O(n\log n)$ operations by computing $\ev_{\calP}(f)$ via FFT. The computation of $\ev_{\calP}(\lambda(x))$, $\ev_{\calP}(\lambda'(x))$, $F(x)$ and $\ev_{\calP}(F'(x))$ need $O(n\log n)$ operations via FFT. Thus the total complexity in step 5 is still $O(n\log n)$.

 \end{proof}
 
 \section{The proof of Lemma~\ref{lem:basis}}\label{app:basis}
\begin{proof}
First, it easy to see that $\cB\subset \cL(\lambda P^{(0)}_{\infty})$ by the definition of Riemann-Roch space. Next, we will show that $\cB$ is $\F_q$-linearly independent. Actually, all $\by^{\bu}x^j\in\cB$ have distinct discrete valuations at $P^{(0)}_{\infty}$, thus they must be linearly independent over $\F_q$. We claim that \[\nu_{P^{(0)}_{\infty}}(\by^{\bu}x^j)= \nu_{P^{(0)}_{\infty}}(\by^{\bu'}x^{j'})\ \text{if\ and\ only\ if}\ u=u' \ \text{and}\ j=j'.\]
This can be proved by induction. Note that if $j\neq j'$, then $\nu_{P^{(0)}_{\infty}}(x^j)\neq \nu_{P^{(0)}_{\infty}}(x^{j'})$. Assume $\nu_{P^{(0)}_{\infty}}(\by^{\bu}x^j)= \nu_{P^{(0)}_{\infty}}(\by^{\bu'}x^{j'})$ and $u_j=u'_j$ for all $j\leq i-1$ in the expansion of $u$ as in Equation~\eqref{eq:expu}. Then
\[\begin{split} &0= \nu_{P^{(0)}_{\infty}}\left(\by^{\bu}x^j\right)-\nu_{P^{(0)}_{\infty}}\left(\by^{\bu'}x^{j'}\right)\equiv \left(\prod_{j=1}^ip_j\right) (u_i-u_i')\nu_{P^{(i)}_{\infty}}(y_i)\bmod \prod_{j=1}^{i+1}p_j. \end{split}\]
Since $\gcd(\nu_{P^{(i)}_{\infty}}(y_i),p_{i+1})=1$, there must be $u_i=u'_i$. By induction, we have $u=u'$. Thus $j=j'$ and the claim is true. 
 
Finally, by the second part in assumption $(P3)$, the set $\cB$ spans the whole Riemann-Roch space $\cL(\lambda P^{(0)}_{\infty})$. We show it by induction. For any $f\in\cL(\lambda P_{\infty})$, since $E=E_1(y_0)$, then $f=\sum_{i=0}^{p_1-1}\alpha_iy_0^i$ where $\alpha_i\in E_{1}$ for all $0\le i\le p_1-1$. Since $\{1, y_0, \dots, y_0^{p_1-1}\}$ is an integral basis of $E/E_1$ at $Q$ for $Q\neq P^{(1)}_{\infty}$, by \cite[Corollary 3.3.5]{sti09}, one has
\[\bigcap_{P\mid Q}\calO_{P}=\bigoplus\limits_{i=0}^{p_1-1} \calO_{Q}y_0^i,\ \forall\  Q\neq P^{(1)}_{\infty}.\]
As $f\in\bigcap_{P\mid Q}\calO_{{P}}$ for $Q\neq P^{(1)}_{\infty}$, by the above equation, the coefficients $\alpha_0,\dots,\alpha_{p_1-1}$ of $f$ are all in $\calO_{Q}$, namely, they have poles only at $P^{(1)}_{\infty}$. For $0\le i\le p_1-1$, assume $\nu_{P^{(1)}_{\infty}}(\alpha_i)=-a_i$. Then $a_i\leq \frac{m(k-1)}{p_1}$ by $\nu_{P^{(0)}_{\infty}}(f)\geq -m(k-1)$. Thus $\alpha_i\in \cL(\frac{m(k-1)}{p_1}P^{(1)}_{\infty})$. By induction on $\alpha_0,\dots,\alpha_{p_1-1}$, we can obtain that
\[f= f_0(x)+f_1(x)\by^{\mathbf 1}+\dots+f_{m-1}(x)\by^{\mathbf{m-1}},\ \text{and}\ f_u(x)\in\cL((k-1)P^{(r)}_{\infty})=\F_q[x]_{\leq k-1}.\] 
Hence $\cB$ is a basis of the Riemann-Roch space $\cL(\lambda P^{(0)}_{\infty})$.
 
\end{proof}

 \section{The proof of Lemma~\ref{lem:gfft}}\label{app:gfft}

\begin{proof}
   Denote the complexity (i.e., the total number of operations in $\F_q$) of evaluating an $f$ at the set $\cP$ by $C(m)$. Let $\calP_i=\{\fP_1,\dots,\fP_m\}\cap E_i$ for $1\leq i\leq r$. Since $P$ is splitting completely in $E/F$ and $[E:E_i]=p_1\cdots p_i$, then $|\calP_i|=m/p_1\cdots p_i$. For any $\fp_i\in\calP_1$, let $\fP_{i,1},\dots,fP_{i,p_1}$ be all $p_1$ places lying over $\fp$. Assume 
   \[f=f_0(y_1,\dots,y_{r-1})+f_0(y_1,\dots,y_{r-1})y+\dots+f_0(y_1,\dots,y_{r-1})y^{p_1-1}.\]
   Then 
   \begin{equation}\label{eq:rec}
\left(\begin{matrix} f(\fP_{i,1})\\
f(\fP_{i,2})\\
\vdots\\
f(\fP_{i,p_1-1})\end{matrix}\right)=\left(\begin{matrix}
1            &            y(\fP_{i,1})    & \cdots  & y^{p_1-1}(\fP_{i,1})\\ 
1 & y(\fP_{i,2}) & \cdots  & y^{p_1-1}(\fP_{i,2}) \\
\vdots       & \vdots       & \cdots   & \vdots \\
1 &   y(\fP_{i,p_{1}}) & \cdots &  y^{p_1-1}(\fP_{i,p_{1}})\end{matrix}\right)\cdot \left(\begin{matrix} f_0(\fp_i)\\
f_1(\fp_i)\\
\vdots\\
f_{p_1-1}(\fp_i)\end{matrix}\right)\end{equation}
   Thus, the computation of $f$ at $m$ points $\fP_1,\dots,\fP_m$ can be reduced to the computation of $f_0,f_1,\dots,f_{p_1-1}$ at $\calP_1$ through the above Equation~\eqref{eq:rec}. Hence $C(m)$ satisfies the following recursive formula
   \[C(m)=p_1C(m/p_1)+p_1^2m/p_1=p_1C(m/p_1)+p_1m.\]
By recursive reduction, we can finally reduce the computation of $(f(\fP_1),\dots,f(\fP_m))$ to MPEs of $m$ functions in $E_r$ at $P$. In this case, by recursive formula of $C(m)$, we have
\[C(m)=p_1\cdots p_rC(1)+(p_1+\dots+p_r)m=O(Bm\log m).\]

Conversely, let $I(m)$ denote the G-IFFT of length, i.e., given the $m$ evaluation of $(f(\fP_1),\dots,f(\fP_m))$, find its $m$ coefficients under basis $\{\by^{\bu}:\ u\in[m]\}$. By the Equation~\eqref{eq:rec}, we have
   \begin{equation}\label{eq:irec}
\left(\begin{matrix} f_0(\fp_i)\\
f_1(\fp_i)\\
\vdots\\
f_{p_1-1}(\fp_i)\end{matrix}\right)=\left(\begin{matrix}
1            &            y(\fP_{i,1})    & \cdots  & y^{p_1-1}(\fP_{i,1})\\ 
1 & y(\fP_{i,2}) & \cdots  & y^{p_1-1}(\fP_{i,2}) \\
\vdots       & \vdots       & \cdots   & \vdots \\
1 &   y(\fP_{i,p_{1}}) & \cdots &  y^{p_1-1}(\fP_{i,p_{1}})\end{matrix}\right)^{-1}\cdot \left(\begin{matrix} f(\fP_{i,1})\\
f(\fP_{i,2})\\
\vdots\\
f(\fP_{i,p_1-1})\end{matrix}\right),\end{equation}
Thus, the computation of $f$ from $m$ evaluations can be reduced to the computation of $f_0,f_1,\dots,f_{p_1-1}$ from $m/p_1$ evaluation through the above Equation~\eqref{eq:irec}. Hence $I(m)$ also satisfies the following recursive formula
   \[I(m)=p_1I(m/p_1)+p_1^2m/p_1=p_1I(m/p_1)+p_1m=O(Bm\log m).\]
\end{proof}

\end{document}